\documentclass[11pt,a4paper]{article}

\RequirePackage[numbers]{natbib}
\RequirePackage[colorlinks,citecolor=blue,urlcolor=blue]{hyperref}

\usepackage{fullpage,lmodern,bm}
\usepackage[T1]{fontenc}
\usepackage[utf8]{inputenc}                        
\usepackage[Lenny]{fncychap}
\usepackage[usenames,dvipsnames]{xcolor}
\usepackage{amsmath,amssymb,amsthm,thmtools,dsfont,enumerate}
\usepackage{tikz,epigraph,lipsum}
\usetikzlibrary{shapes,arrows}
\usepackage{rotating,lipsum,pifont,cases}
\usepackage{appendix,authblk}
\usepackage{cases}
\usepackage{pdfsync}
\usetikzlibrary{patterns,intersections}
\definecolor{blue0}{RGB}{0,77,153} 
\definecolor{red0}{RGB}{179,0,77} 
\definecolor{green0}{RGB}{134,219,76} 
\definecolor{gray0}{RGB}{84,97,110}
\usepackage{capt-of}
\usepackage{pifont}
\newcommand{\cmark}{\ding{51}}%
\newcommand{\xmark}{\ding{55}}%
\usepackage{float} 

\numberwithin{equation}{section}

\newtheorem{theorem}{Theorem}[section]

\newtheorem{lemma}[theorem]{Lemma}

\newtheorem{remark}[theorem]{Remark}

\numberwithin{equation}{section}
\def\vs#1{\vspace{#1mm}}

\def\be{\begin{align}}
\def\ee{\end{align}}
\def\b*{\begin{eqnarray*}}
\def\e*{\end{eqnarray*}}

\def\be{\begin{eqnarray}}
\def\ee{\end{eqnarray}}
\def\beq{\begin{equation}}
\def\eeq{\end{equation}}
\def\b*{\begin{eqnarray*}}
\def\e*{\end{eqnarray*}}
\def\bi{\begin{itemize}}
\def\ei{\end{itemize}}

\def \1{{\bf 1}}

\def\={\;=\;}

\def\argmin{\mbox{\rm arg}\min}

 \def\vs#1{\vspace{#1mm}}

\def \E{\mathbb{E}}
\def \F{\mathbb{F}}

\def \Q{\mathbb{Q}}
\def \R{\mathbb{R}}

\def\Fc{{\cal F}}
\def\Gc{{\cal G}}

\newcommand{\Mid}{{\ \Big|\ }}

\definecolor{blue0}{RGB}{0,77,153} 
\definecolor{red0}{RGB}{179,0,77} 
\definecolor{green0}{RGB}{134,219,76} 
\definecolor{gray0}{RGB}{84,97,110}

\title{{Lifting the Heston model}}	
\date{\today}

\author{Eduardo {Abi Jaber}\thanks{abijaber@ceremade.dauphine.fr.  {I would like to thank Bruno Bouchard, Camille Illand, Mathieu Rosenbaum and Sergio Pulido for very fruitful discussions and insightful comments. I would also like to thank two anonymous referees for their  careful reading and their suggestions.}}}

\affil{AXA Investment Managers,  Multi Asset Client Solutions, Quantitative Research, \break
                        6 place de la Pyramide, 92908 Paris - La Défense, France.}
\affil{Universit\'e Paris-Dauphine, PSL University, CNRS, CEREMADE, 75016 Paris, France.}

\begin{document}

	\maketitle

	\begin{abstract}

How to reconcile the classical Heston model with its rough counterpart? We introduce a lifted version of the Heston model with $n$ {multi-factors, sharing} the same Brownian motion but mean reverting at different speeds. Our model nests as extreme cases the classical Heston model (when $n = 1$), and the rough Heston model (when $n$ goes to infinity). We show that the lifted model enjoys the best of both worlds: Markovianity, satisfactory fits of implied volatility smiles for short maturities with very few parameters, and consistency with the statistical roughness of the realized volatility time series. Further, our approach speeds up the calibration time and opens the door to time-efficient simulation schemes.\\
		
		\noindent  {{\textit{Keywords:}} Stochastic volatility, implied volatility, affine Volterra processes,  Riccati equations,  rough volatility}.  \vs1 
		
		
	\end{abstract}



\section{Introduction}
Conventional one-dimensional continuous stochastic volatility models, including the renowned Heston model \cite{heston1993closed}:
\begin{align}
dS_t &= S_t \sqrt{V_t}dB_t, \quad S_0>0, \label{eq:Heston S}\\
dV_t &= \lambda (\theta - V_t) dt + \nu \sqrt{V_t}dW_t, \quad V_0 \geq 0, \label{eq:Heston V}
\end{align}
 have struggled  in capturing the risk of large price movements on a short timescale. In the pricing world, this translates into failure to reproduce the at-the-money {skew of the implied volatility observed in the market as illustrated on Figure \ref{fig:skew intro} below}.

\begin{center}
	\includegraphics[scale=0.6]{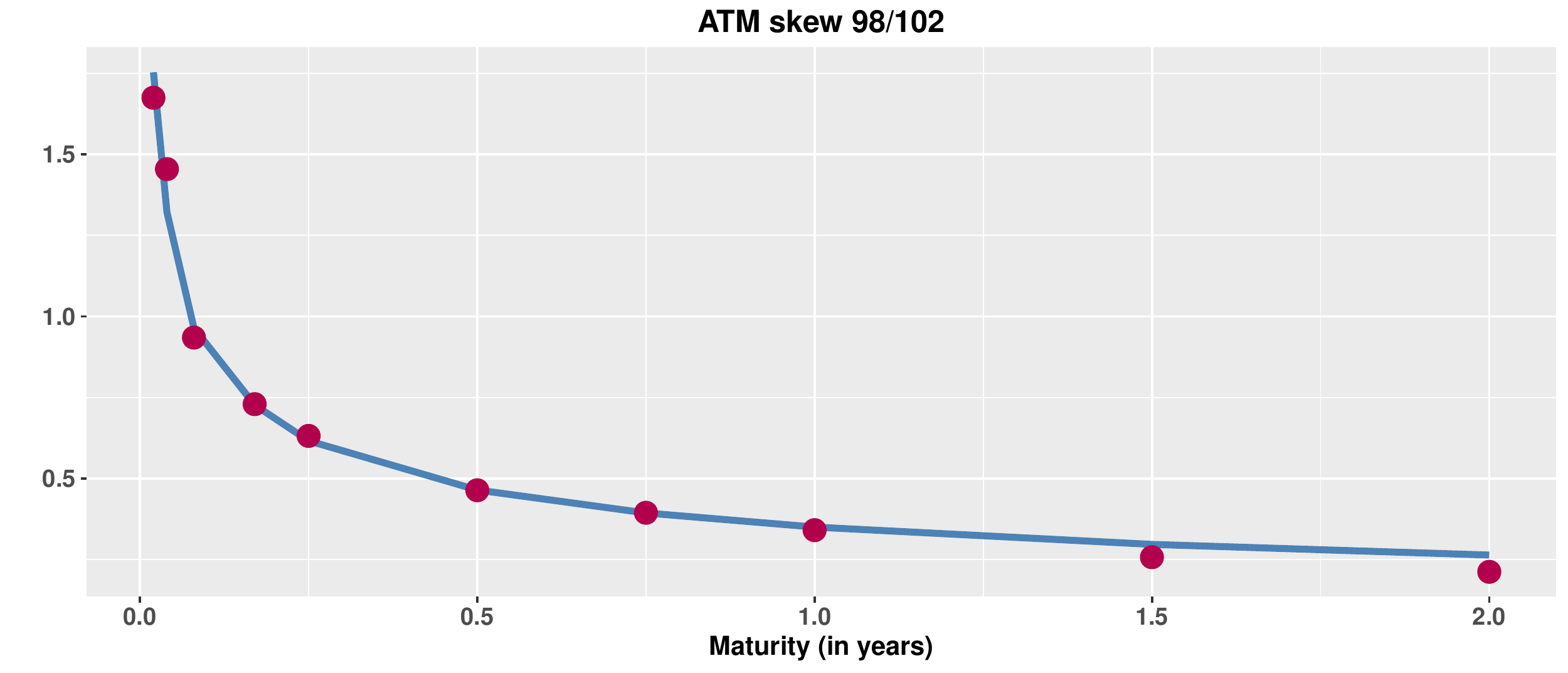}
	\captionof{figure}{{Term structure of the at-the-money skew of the implied volatility $\frac{\partial \sigma_{\text{implicit}}(k,T)}{\partial k}\big |_{k=0} $ for the S\&P index on June 20, 2018 (red dots) and a power-law fit $t \to 0.35\times t^{-0.41}$. Here $k:=\ln(K/S_0)$ stands for the log-moneyness and $T$ for the time to maturity. {(The appellation \textit{skew} is justified by the following relation $ \frac{\partial \sigma_{\text{implicit}}(k,T)}{\partial k}\big |_{k=0} \approx \frac{s}{6 \sqrt{T}}, $ where $s$ is the skew of $\log S_T$, see \cite[(5.93) on p.~194]{B15}.)}  }}
	\label{fig:skew intro}
\end{center}

In view of improving the overall fit,  several  directions have been considered over the past decades. Two of the  most common extensions are   adding   jumps \cite{conttankov,gatheral2011volatility}  and stacking additional random factors \cite{bergomi2005smile,fouque2011multiscale}, in order to  jointly account for short and  long timescales. While the two approaches have structural differences, they  both  suffer, {in general},  from the curse of dimensionality, as more   parameters are introduced, slowing down the calibration process {(one notable exception is the Variance-Gamma model  \cite{madan1998variance}).} Recently, rough volatility models  have been introduced as {a  fresh substitute with remarkable fits of the implied volatility surface, see \cite{Bayeretal2016,euch2017roughening,gatheral2014volatility}.} {The  rough variance process involves a  one-dimensional Brownian motion,  keeps the number of parameters small and enjoys continuous paths.}  {However, the price to pay is that rough volatility models leave the realm of semimartingale and Markovian models, which makes pricing and hedging  a challenging task, while degrading the calibration time.} Here, the curse of dimensionality hits us straight in the face in the non-Markovianity of the process. {Indeed, the rough model can be seen as an infinite dimensional Markovian model, as shown in \cite{AJEE18b,cuchiero2018generalized}.} \vs2

Going back to the standard Heston model \eqref{eq:Heston S}-\eqref{eq:Heston V}, despite its lack of fit for short maturities, {it  remains increasingly popular among practitioners. This is due to its high tractability, by virtue of the closed form solution of the characteristic function, allowing for fast pricing and calibration by Fourier inversion techniques \cite{carr1999option,fang2008novel}.} Recently, El Euch and Rosenbaum \cite{euch2016characteristic}  combined the tractability of the Heston model with the flexibility {of rough volatility models,   to elegantly concoct a rough counterpart of \eqref{eq:Heston S}-\eqref{eq:Heston V}, dubbed  the rough Heston model.} More precisely, the rough model is constructed by replacing the variance process \eqref{eq:Heston V} by a fractional square-root process as follows
				\begin{align}
				dS_t &= S_t \sqrt{V_t}dB_t, \quad S_0>0,  \label{eq:rough price}\\
			 V_t &= V_0 +  \frac{1}{\Gamma(H+1/2)}\int_0^t (t-s)^{H - 1/2} \left( \lambda( \theta - V_s)ds  + \nu\sqrt{V_s}dW_s\right), \label{eq:rough variance}
				\end{align} 
{where $H \in (0,1/2]$   has a physical interpretation, as it measures the regularity of the sample paths of $V$, see \cite{BLP:16, gatheral2014volatility}, the case $H=1/2$ corresponding to the standard Heston model. More precisely, the sample paths of $V$ are locally H\"older continuous of any order strictly less than $H$.}   As for the standard Heston model, the characteristic function of the log-price is known, but only up to the solution of a certain fractional Riccati Volterra equation. {Indeed, both models belong to the tractable and unifying class of affine Volterra processes introduced in \cite{ALP17}.} {The following table summarizes the characteristics of the two models.} \vs2\vs2\vs2\vs2\vs2
		
\begin{table}[h!]
	\centering  
	\begin{tabular}{c c  c } 
		
		\hline
		Characteristics & Heston & Rough Heston  \\
		\hline \hline                 
		Markovian & \textcolor{blue0}{\cmark} & \textcolor{red0}{\xmark}   \\
		Semimartingale & \textcolor{blue0}{\cmark} & \textcolor{red0}{\xmark}  \\
		Simulation & \textcolor{blue0}{Fast} & \textcolor{red0}{Slow}  \\
		&&\\
		Affine Volterra process &\textcolor{blue0}{\cmark} &\textcolor{blue0}{\cmark} \\
		Characteristic function & \textcolor{blue0}{Closed} & \textcolor{red0}{Fractional Riccati} \\
		Calibration  & \textcolor{blue0}{Fast} & \textcolor{red0}{Slower} \\
		&&\\
		Fit short maturities & \textcolor{red0}{\xmark} & \textcolor{blue0}{\cmark} \\
		Regularity of sample paths & \textcolor{red0}{$H=0.5$} & \textcolor{blue0}{$0<H\leq 0.5$} \\
		\hline 
	\end{tabular}
	\caption{Summary of the characteristics of the models.}
	\label{tablecomparisonintro} 
\end{table}

{
In the present paper, we study a conventional multi-factor  continuous stochastic volatility  model: \textit{the lifted Heston model}. {The variance process is constructed  as a weighted sum of $n$ factors, driven by the same one-dimensional Brownian motion, but mean reverting at different speeds, in order to accommodate a full spectrum of timescales. At first glance, the model seems over-parametrized, with already $2n$  parameters for the mean reversions and the weights. Inspired by the approximation results of \cite{AJEE18a}}, we provide a good parametrization of these $2n$ parameters in terms of one single parameter $H$, which is nothing else {but} the Hurst index of a limiting rough {Heston}  model \eqref{eq:rough price}-\eqref{eq:rough variance}, obtained after  sending the numbers of factors to infinity.}\vs2

{The lifted  model not only nests as extreme cases the classical Heston model (when $n = 1$) and the rough Heston model (when $n$ goes to infinity), but also   enjoys the best of both worlds: the flexibility of  rough volatility models, and the Markovianity of  their conventional counterparts.} Further, the model remains tractable{, as it also belongs to the class of affine Volterra processes.}  Here, the characteristic  function of the log-price is known up to a solution of a finite system of Riccati ordinary differential equations.  From a practical viewpoint, we demonstrate that the \textit{lifted Heston model}: 
\begin{itemize}
	\item
	reproduces the same volatility surface as the rough Heston model for maturities ranging from one week to two years, 
	\item
	mimics the explosion of the  at-the-money skew for short maturities,
	\item
	calibrates twenty times faster than its rough counterpart,
	\item
	is easier to simulate than the rough model,
	\item
	tricks the human eye as well as  statistical estimators of the Hurst index.
\end{itemize}
All in all,  the \textit{lifted Heston model} can be  more easily implemented than its rough counterpart, while still retaining the precision of implied volatility fits  of the rough Heston model. Further, the \textit{lifted Heston model} is able to generate a volatility surface, which cannot be  generated by the classical Heston model, with only one additional parameter. The lifted \textit{lifted Heston model} is also consistent with the statistical roughness of realized volatility times series across different timescales. Finally, the stock price and the variance process  enjoy continuous paths and only depend on a two-dimensional Brownian motion, leading to simple and feasible hedging strategies.\vs2

{The \textit{lifted Heston model} appeared for the first time in \cite{AJEE18a} as a multi-factor approximation of the rough Heston model, with hundreds of factors. In the present paper, we take the \textit{lifted Heston model} as our starting model and we argue that few factors are sufficient in practice. In addition, we provide a thorough  numerical study for calibration, robustness, simulation and estimation. This constitutes a crucial step towards the implementation of rough volatility models in practice that can be easily extended to other models than the Heston model. We  mention  \cite{bennedsen2017hybrid, gatheral2018rational, giorgia2018rough, HJM17} for several numerical algorithms for rough volatility models.}\vs2

The paper is outlined as follows. In Section \ref{S:lifted heston} we introduce our \textit{lifted Heston model} and provide  its existence, uniqueness and its affine Fourier-Laplace transform. Exploiting the limiting rough  model, we proceed in Section \ref{S:parameter reduction} to a reduction of the  number of parameters to calibrate.  {Numerical experiments for the model, with $n=20$ factors, are illustrated  in Section \ref{S:numerics}, for {calibration}, simulation and statistical estimation of the roughness. Finally, some technical  material is postponed to  Appendices \ref{eq:appendix existence}-\ref{A:scheme}.}

\section{The lifted Heston model}\label{S:lifted heston}

We fix $n \in \mathbb N$ and we define the \textit{lifted Heston model} as a conventional  stochastic volatility model,  with $n$ factors for the variance process: 
\begin{align}
dS_t^n &= S_t^n \sqrt{V^n_t} dB_t, \quad S_0^n>0,  \label{eq:liftedS}\\
V^n_t &= g_0^n(t) + \sum_{i=1}^n c^n_i U^{n,i}_t, \label{eq:liftedV}\\
dU_{t}^{n,i} &= \left( -x^n_i U^{n,i}_t -\lambda V^n_t \right) dt + \nu \sqrt{V^n_t} dW_t, \quad U^{n,i}_0 =0, \quad i=1,\ldots,n, \label{eq:liftedU}
\end{align}
 with parameters the function $g_0^n$, $\lambda,\nu \in\R_+ $,  $c^n_i,x^n_i \geq  0$, for $i=1,\ldots,n$, and $B=\rho W + \sqrt{1-\rho^2}W^{\perp}$, with $(W,W^{\perp})$ a two dimensional Brownian motion on a fixed filtered probability space $(\Omega, \Fc, \F:=(\Fc_t)_{t \geq 0}, \Q)$, with $\rho \in [-1,1]$. \vs2

We stress that all the factors $(U^{n,i})_{1 \leq i \leq n}$  start from zero\footnote{Notice that the initial value of the variance process  $V^n$  is $g_0^n(0)$.} and  share the same dynamics, with the same one-dimensional Brownian motion $W$, except that they mean revert at different speeds $(x_i^n)_{1 \leq i \leq n}$. Further, the deterministic input curve $g_0^n$ allows one to plug-in initial term-structure curves. More precisely, taking the expectation in \eqref{eq:liftedV} leads to the following relation 
$$ \E[V^n_t] + \lambda \sum_{i=1}^n c^n_i \int_0^t e^{-x^n_i (t-s)} \E[V^n_s] ds = g_0^n(t), \quad  t \geq 0. $$
{In practice, the forward variance curve, up to a horizon $T>0$,  can be extracted from variance swaps observed in  the market and then plugged-in in place of $(\E[V^n_t])_{t \leq T}$ in the previous expression.}\vs2
 For a suitable choice of continuous curves $g_0^n$, for instance 
if 
\begin{align}\label{eq:g0 example 1}
g^n_0 \mbox{ is non-decreasing such that } g^n_0(0) \geq 0,
\end{align} 
or 
\begin{align}\label{eq:g0 example 2}
 g_0^n:t\to V_0 + \sum_{i=1}^n  c_i^n  \int_0^t e^{-x_i^n(t-s)} \theta (s) ds, \mbox{ with } V_0 , \theta \geq 0, 
 \end{align} 
there exists a unique continuous $\F$-adapted strong solution $(S^n,V^n,(U^{n,i})_{1 \leq i \leq n})$ to \eqref{eq:liftedS}-\eqref{eq:liftedU}, such that $V^n_t \geq 0$, for all $t \geq 0$, and $S^n$ is a $\F$-martingale. {We refer to Appendix \ref{eq:appendix existence} for more details and the exact definition of the set of admissible input curves $g_0^n$.} \vs2

Since {our main objective is  to compare the lifted model  to other existent models, we will restrict to the case of input curves of the form} 
\begin{align}\label{eq: flat curve}
g_0^n:t\to V_0 + \lambda \theta \sum_{i=1}^n  c_i^n  \int_0^t e^{-x_i^n(t-s)}   ds , \quad \mbox{with } V_0,\theta \geq 0 .
\end{align}
Setting $n=1$, $c^1_1= 1$ and $x^1_1=0$, the \textit{lifted Heston model} degenerates into the standard Heston  model \eqref{eq:Heston S}-\eqref{eq:Heston V}.  So far, the multi-factor extensions of the standard Heston model have been considered  by stacking additional  square-root processes as in the double Heston model\footnote{The double Heston model is defined  in   \cite{christoffersen2009shape} as follows 
\begin{align}
dS_t &= S_t\left( \sqrt{U^1_t} dB^1_t + \sqrt{U_t^2} dB^{2}_t\right), \label{eq:doubleS}\\
dU^i_t &= \lambda_i(\theta_i- U^i_t )dt + \nu_i   \sqrt{U^i_t} dW^i_t, \quad  U^i_0 \geq 0, \quad i \in \{1,2\} \label{eq:doubleU}, 
\end{align}
where $B^i = \rho_i W^i + \sqrt{1-\rho^2_i} W^{i,\perp} $ with $\rho_i \in [-1,1]$ and  $(W^1,W^2,W^{1,\perp},W^{2,\perp})$ a four-dimensional Brownian motion.}  of \cite{christoffersen2009shape} and the multi-scale Heston model of \cite{fouque2011fast}, or by considering a Wishart matrix-valued process as in \cite{da2008multifactor}.  {In both cases,  the dimension of the driving Brownian motion  for the variance process, along with the number of parameters, grows with the number of factors.} Clearly,  the  \textit{lifted Heston model}   differs from these extensions, {one can compare \eqref{eq:liftedS}-\eqref{eq:liftedU} for $n=2$ with  \eqref{eq:doubleS}-\eqref{eq:doubleU}.}  \vs2

{Just like the classical Heston model, the \textit{lifted Heston model} remains tractable. Specifically, fix $u \in \mathbb C$ such that $\mbox{Re}(u) \in [0,1]$.} {By virtue of Appendix \ref{S: appendix full fourier},} the Fourier-Laplace transform of the log-price is exponentially affine with respect to the factors $(U^{n,i})_{ 1 \leq i \leq n}$:
		\begin{align}\label{eq:char function log S}
		\E\left[ \exp\left(u \log S^n_t \right) \Mid \Fc_t \right] = \exp\left({\phi^n(t,T)  +  u \log S^n_t  +  \sum_{i=1}^n  c^n_i\psi^{n,i} (T-t) U^{n,i}_t }\right),
		\end{align}
		for all $t \leq T$, where   $(\psi^{n,i})_{1 \leq i \leq n}$ solves the following  $n$-dimensional system of Riccati ordinary differential equations
\begin{align}\label{eq:psi xi}
(\psi^{n,i})' = - x^n_i  \psi^{n,i} + F\left(u,   \sum_{j=1}^n c^n_j \psi^{n,j} \right),\quad  \psi^{n,i}(0)=0, \quad i=1,\ldots, n,
\end{align}
with 
\begin{align}\label{eq:RiccatiF}
F(u,v) = \frac{1}{2}(u^2 - u) + (\rho \nu u - \lambda) v + \frac{\nu^2}{2} v^2, 
\end{align} 
and 
 $$\phi^n(t,T) =   \int_0^{T-t} F\left(u,\sum_{i=1}^n  c^n_i\psi^{n,i}(s) \right)g^n_0 (T-s) ds, \quad t \leq T.$$

In particular, for $t=0$, since $U^{n,i}_0=0$ for $i=1,\ldots, n$, the unconditional Fourier-Laplace transform reads
		\begin{align}\label{eq:char function log S t=0}
		\E\left[ \exp\left(u \log S^n_t \right)\right] = \exp\left( {u \log S^n_0} +  \int_0^{T} F\left(u,\sum_{i=1}^n  c^n_i\psi^{n,i}(s) \right)g^n_0 (T-s) ds \right).
		\end{align}

A similar formula holds for  the Fourier-Laplace transform of the joint process $(\log S^n, V^n)$ with  integrated log-price and variance, we refer to the Appendix \ref{S: appendix full fourier} for the precise expression.\vs2

Consequently, the Fourier-Laplace transform of the \textit{lifted Heston model} is known in closed-form, up to the solution of a deterministic $n$-dimensional system of ordinary differential equations \eqref{eq:psi xi}, which can be solved numerically. Once there, standard Fourier inversion techniques can be applied  on \eqref{eq:char function log S t=0} to deduce option prices.     This is illustrated in the following sections.

\section{Parameter reduction and the choice of the number of factors}\label{S:parameter reduction}

{In this section, we proceed to a reduction of the number of parameters to calibrate. Our  inspiration stems from  rough volatility. In a first step, for every $n$, we  provide a parametrization of the weights and the mean reversions $(c^n_i,x^n_i)_{1 \leq i \leq n}$ in terms of the Hurst index $H$ of a limiting rough volatility model and one additional parameter $r_n$.  Then, we specify the number of factors $n$ and the value of the additional parameter $r_n$ so that the lifted model  		reproduces the same volatility surface as the rough Heston model for maturities ranging from one week up to two years,	while calibrating twenty times faster than its rough counterpart.  Benchmarking  against rough volatility models is justified by the fact that  one of the main strengths of these models is their ability  to achieve better fits of the implied volatility surface than conventional one-dimensional stochastic volatility models. This has been illustrated on real market data in \cite{Bayeretal2016,euch2017roughening}. Finally, for the sake of completeness, we provide a comparison with the standard Heston model.
}
\subsection{Parametrization in terms of the Hurst index}
{For {an initial input curve of the form \eqref{eq: flat curve}, the \textit{lifted Heston model} \eqref{eq:liftedS}-\eqref{eq:liftedU} has the same five parameters $(V_0, \theta, \lambda, \nu, \rho)$ of the Heston model,   plus  $2n$ additional  parameters for the weights and the  mean reversions $(c^n_i,x^n_i)_{1 \leq i \leq n}$.\footnote{{{If one chooses $g^n_0$ to match the 
 forward variance curve, then, the parameters $(V_0,\theta)$ can be eliminated from both models.}}}} At first sight, the model seems to suffer from the curse of dimensionality,  as it requires the calibration of $(2n+5)$ parameters.  This is where the exciting theory of  rough volatility finally  comes   into play. Inspired by the approximation result \cite[Theorem 3.5]{AJEE18a}, we suggest to use  a parametrization of $(c^n_i,x^n_i)_{1 \leq i \leq n}$ in terms of two well-chosen parameter. By doing so, we reduce the $2n$ additional parameters to calibrate to only two effective parameters.} \vs2

{Qualitatively, we choose the weights and mean reversions $(c^n_i,x^n_i)_{1 \leq i \leq n}$ in such a way that sending the number of factors $n  \to \infty$ would yield the convergence of the \textit{lifted Heston model}  towards a rough Heston model \eqref{eq:rough price}-\eqref{eq:rough variance},  with parameters $(V_0, \theta, \lambda, \nu, \rho,H)$. The additional parameter $H \in (0,1/2)$ is the so-called Hurst index of the limiting fractional variance process \eqref{eq:rough variance}, and it measures the regularity of its sample paths.} {This is possible by virtue of an infinite-dimensional  Markovian representation of the limiting rough variance process \eqref{eq:rough variance} due to \cite{AJEE18b}, which we recall in the following remark.}\vs2

\begin{remark}[Representation of the limiting rough process]\label{R:representation}
	The fractional kernel 
	appearing in  the limiting  {rough process} \eqref{eq:rough variance}  admits the following {Laplace representation}
	\begin{align*}
	\frac{t^{H- 1/2}}{\Gamma(H+1/2)}	=\int_0^{\infty} e^{-xt} {\mu(dx)}, \quad \mbox{with  }\; {\mu(dx) = \frac{x^{- H - 1/2}}{\Gamma(1/2-H) \Gamma(H+1/2)}},
	\end{align*}	
	so that the stochastic Fubini theorem, after setting $V_0\equiv0$ in \eqref{eq:rough variance}, leads to
	\begin{align*}
	V_t &= \int_0^{\infty} {U_t(x)}  {\mu(dx)}, \quad x >0,
	\end{align*}	 
	where, for all $x>0$,
	\begin{align*}
	{U_t(x):=\int_0^t e^{-x(t-s)} \left(\lambda ( \theta - V_s) ds +  \nu  \sqrt{V_s} dW_s \right)}.
	\end{align*}	 
	This can be seen as the {mild formulation} of the following stochastic partial differential equation
	\begin{align}
	d{U_t(x)}&= \left({-xU_t(x)} + \lambda \left( \theta -\int_0^{\infty} U_t(y) \mu(dy)  \right)\right) dt +  \nu  \sqrt{\int_0^{\infty} U_t(y) \mu(dy)} dW_t, \label{eq:spde U1}\\
	U_0(x)&=0, \quad x>0. \label{eq:spde U2}
	\end{align}	 
	Whence, the {rough process} can be reinterpreted as a superposition of infinitely many factors ${(U_{\cdot}(x))_{x>0}}$ sharing the same dynamics but mean reverting at {{different speeds $x \in (0,\infty)$}.} We refer to \cite{AJEE18b} for the rigorous treatment of this representation. One makes the following observations:
	\begin{itemize}
		\item 
		{multiple timescales} are naturally encoded in rough volatility models, which can be a plausible explanation for  their ability to achieve better fits than conventional one-dimensional models,
		\item 
		the largest mean reversions going to infinity characterize the factors responsible of the {roughness} of the process.
	\end{itemize}
\end{remark}

More precisely,  for a fixed even number of factors $n$, \eqref{eq:liftedU} corresponds to a discretization of \eqref{eq:spde U1} in the $x$-variable, after approximating $\mu$ by a sum of diracs $\sum_{i=1}^n c_i^n \delta_{x_i^n}$. We fix $r_n>1$ and we consider the following parametrization for  the weights and the mean reversions
 {\begin{align}\label{eq: ci and xi}
 c^n_i=\frac{(r_n^{1-\alpha}-1)r_n^{(\alpha-1)(1+n/2)}}{\Gamma(\alpha)\Gamma(2-\alpha)} r_n^{(1-\alpha)i} \;\;\mbox{ and }\;\;  x^n_i= \frac{1-\alpha}{2-\alpha}\frac{r_n^{2-\alpha}-1}{r_n^{1-\alpha}-1} r_n^{i-1-n/2},  \;\;  i=1,\ldots,n,
 \end{align}} 
 where $\alpha:=H+1/2$ for some $H\in (0,1/2)$.\footnote{{{This corresponds to equation (3.6) in \cite{AJEE18a} with the geometric partition  $\eta^n_i=r_n^{i-n/2}$ for $i=0,\dots,n$, which is in the spirit of \cite{CCM00} for the approximation of the factional Brownian motion.}}}\vs2

 {If in addition, the sequence $(r_n)_{n \geq 1}$ satisfies 
 	\begin{align}\label{eq:rn cond}
 	r_n \downarrow 1 \quad \mbox{and} \quad  n \ln r_n \to \infty, \quad \mbox{as } n\to \infty,
 	\end{align}
 then, Theorem \ref{T:convergence rHeston}  in the Appendix ensures the convergence of the lifted model towards the rough Heston model, as $n$ goes to infinity. We refer to Appendix \ref{A:limiting model} for more details.}\vs2

{In order to visualize this convergence,   we first   generate our benchmark implied volatility surface, 
\begin{align}
&\mbox{for $9$ maturities $T \in \{\mbox{1w, 1m, 2m, 3m, 6m, 9m, 1y, 1.5y, 2y}\}$},  \label{eq:T}\\
&\quad\quad\quad\quad\quad \mbox{with up to $80$ strikes $K$ per maturity,} \label{eq:K}
\end{align} 
with a rough Heston model with parameters $\Theta_0:=(V_0,\theta,\lambda,\nu,\rho,H)$ given by
\begin{align}\label{eq: 5parameters}                                                                                                                                                                                                                                                                                                                                                                                                                                                                                                                                                                                                                                                                                                                                                           
V_0=0.02, \quad \theta= 0.02, \quad \lambda = 0.3, \quad   \nu=0.3, \quad \rho=-0.7 \quad \mbox{and} \quad H=0.1.                                                                                                                                                                                                                                                                                                                                                                                                                                                                                                                                                                                                                                                                                                          
\end{align}      

{We recall that the implied volatility surface can be computed by Fourier inversion techniques. Indeed,  it follows from \cite{ALP17,euch2016characteristic} that the Fourier-Laplace transform of the log-price in the rough Heston model \eqref{eq:rough price}-\eqref{eq:rough variance} is of the form
$$ \E[\exp\left(u \log S_T\right)] = \exp\left( u \log S_0 +   \int_0^T F(u,\psi(s,u)) g_0(T-s)ds \right),$$
where $F$ is given by \eqref{eq:RiccatiF}, 
$$  g_0(t)=V_0 + \lambda \theta \int_0^t \frac{s^{H-1/2}}{\Gamma(H+1/2)}ds,$$
and  $\psi$ solves the following fractional Riccati
equation
\begin{align}\label{eq:RiccatiFrac}
 \psi(t,u) =  \frac{1}{\Gamma(H+1/2)}\int_0^t (t-s)^{H - 1/2} F(u,\psi(s,u))ds.
\end{align}		
One then solves \eqref{eq:RiccatiFrac}  numerically  and computes the implied volatilities by Fourier inversion techniques. Here the Adams Predictor-Corrector scheme \cite{diethelm2002predictor} is used with 200 time steps for the discretization of the fractional Riccati equation \eqref{eq:RiccatiFrac}, we refer to \cite[Appendix A]{euch2017roughening} for a complete exposition of this discretization scheme. Then,  call prices are computed via the cosine method \cite{fang2008novel} for the inversion of the characteristic function.\footnote{We note that other Fourier inversion techniques can be used for the second step, for instance, the Carr-Madan method  \cite{carr1999option}, as done in  \cite{euch2017roughening}. As illustrated in \cite{fang2008novel}, for the same level of accuracy, the cosine method is approximately 20 times  faster than the Carr-Madan method, and needs drastically less evaluation points  of the characteristic function $(E\left[\exp\left(u_i \log S^n_t\right)\right])_{i \in \mathcal I}$ ($|\mathcal I|=160$ for the cosine methods and $|\mathcal I|=4096$ for the Carr-Madan method). This latter point is crucial in our case since, for every $i \in \mathcal I$, evaluation of $E\left[\exp\left(u_i \log S^n_t\right)\right]$ requires a numerical discretization of the corresponding Riccati equation.} 
} The generated implied volatility  is kept fixed and is denoted by $\sigma_{\infty}(K,T;\Theta_0)$, for every pair $(K,T)$ in \eqref{eq:T}-\eqref{eq:K}. \vs2      

Then, we define the following sequence 
\begin{align}\label{eq:r_n choice 1}
r_n =1+ 10\,n^{-0.9}, \quad n \geq 1,
\end{align}
which clearly satisfies \eqref{eq:rn cond}.}  For each $n \in \{10,20,50,100,500\}$, we generate the implied volatility surface of  the \textit{lifted Heston model}\footnote{The implied volatility surface is generated by first   solving numerically the $n$-dimensional Riccati equations \eqref{eq:psi xi} with the explicit-implicit scheme \eqref{eq:explicit implicit riccati} detailed in the Appendix with a number of time steps $N=300$.  As before, the  call prices are then  computed via the cosine method \cite{fang2008novel} for the inversion of the characteristic function.} with $n$-factors, with the same set of parameters $\Theta_0$ as in \eqref{eq: 5parameters}, and \eqref{eq:r_n choice 1} plugged in \eqref{eq: ci and xi}. For each $n$, the generated surface is denoted by $\sigma_n(K,T;r_n,\Theta_0)$, for every pair $(K,T)$ in \eqref{eq:T}-\eqref{eq:K}.  \vs2

Because the sequence $(r_n)_{n \geq 1}$ defined in \eqref{eq:r_n choice 1}  satisfies condition \eqref{eq:rn cond}, as $n$ grows, 
$$ \sigma_{n}(K,T;r_n,\Theta_0)  \to \sigma_{\infty}(K, T;\Theta_0),$$
by virtue of Theorem \ref{T:convergence rHeston} in the Appendix.
This convergence phenomenon is illustrated on Figure \ref{fig:convergence factors} below for two maturity slices, one week and one year.

\begin{center}
	\includegraphics[scale=0.65]{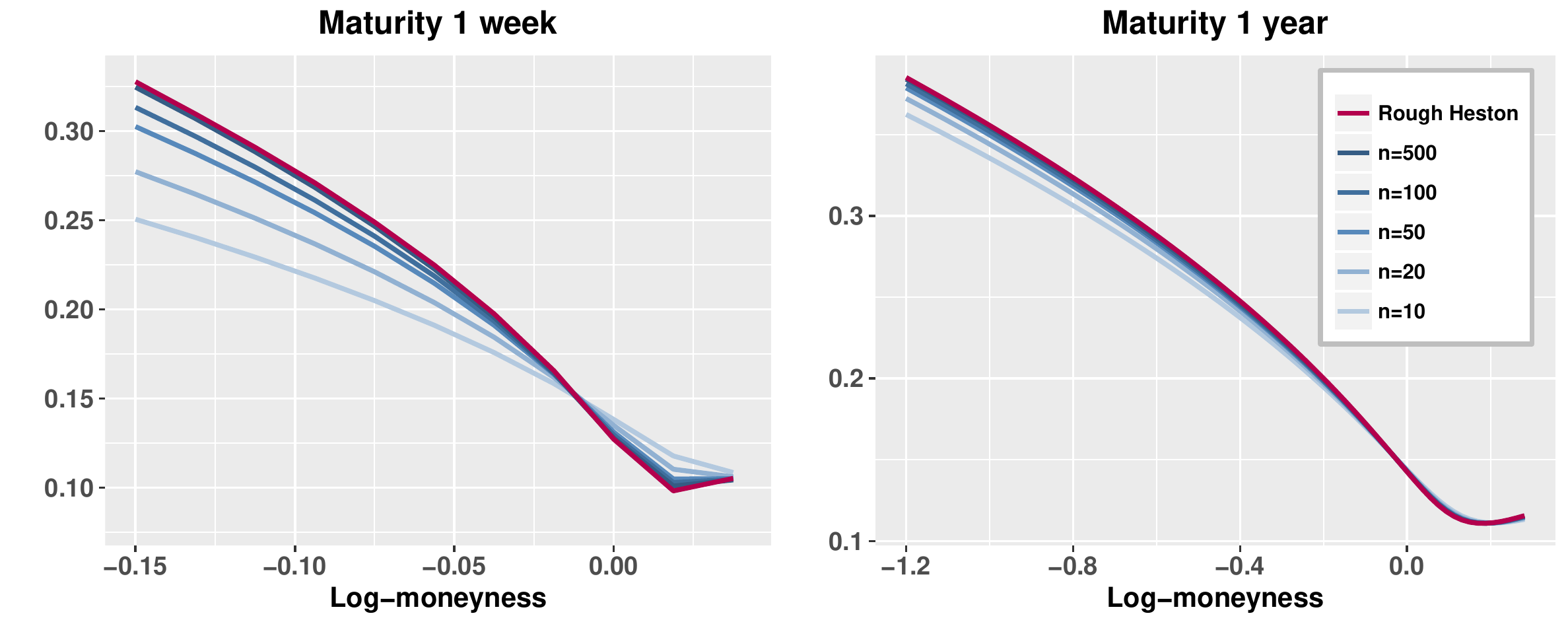}
	\rule{35em}{0.5pt}
	\captionof{figure}{{Convergence of the implied volatility surface of the lifted model $\sigma_n(k,T;r_n,\Theta_0)$, with $r_n=1+10\, n^{-0.9}$, towards its rough counterpart $\sigma_{\infty}(k,T;\Theta_0)$, illustrated on two maturities slices $T \in \{1 \mbox{ week}, 1 \mbox{ year}\}$. Here $k:=\ln(K/S_0)$ stands for the log-moneyness.}}
	\label{fig:convergence factors}
\end{center}

In view of assessing the proximity between the implied volatility surface $\sigma_{n}(K,T;r_n,\Theta_0)$ of the \textit{lifted Heston model} and that of the {rough Heston model} $\sigma_{\infty}(K,T;\Theta_0)$, we compute the mean squared error (MSE) between the two volatility surfaces defined as follows
\begin{align*}
\frac{1}{\sum_{(K',T')} w(K',T')} \sum_{(K,T)} w(K,T) (\sigma_n(K,T;r_n,\Theta_0) - \sigma_{\infty}(K,T;\Theta_0))^2,
\end{align*}
where we sum over all pairs $(K,T)$ as in \eqref{eq:T}-\eqref{eq:K}. Here, $w$ stands for a matrix of  weights, where we put more weight on options near the money and with  short time to maturity (one could also set $w(K,T)=1$ for all $(K,T)$).\vs2

The corresponding mean squared errors of Figure \ref{fig:convergence factors} are reported in Table \ref{tableconvergence} below, along with the computational time\footnote{{All cpu times are computed on a  laptop with  Intel core i7 processor  at 2.2GHz  and 16GB of memory. The code, written in R, is far from being optimized.}} for generating the whole volatility surface, for all pairs $(K,T)$ as in \eqref{eq:T}-\eqref{eq:K}, that is, for $9$ maturities slices with up to $80$ strikes per maturity.\footnote{{One cannot draw definite quantitative conclusions regarding the  comparison between the computational times of the lifted surface and  the one of  the rough surface. Indeed, one needs a more careful study of the discretization errors of the corresponding Riccati equations before comparing the computational times needed to reach the same level of accuracy. We omit to do so here.  However, even if one reduces the number of time steps from 200 to 150 in the Adams scheme, it still takes 67.2 seconds to compute the rough surface. Recall that we used $N=300$ time steps for the $n$-dimensional Riccati equation of the lifted model. In any case, it should be clear that solving the $20$-dimensional Riccati equations is considerably faster then solving the fractional Riccati equation.}}
\begin{table}[H]
	\centering  
	\begin{tabular}{c c  c cc} 
		\hline
		& $n$ & $r_n=1+10\,n^{-0.9}$ & Time (seconds) & MSE \\ 
		\hline \hline                 
		\textcolor{blue0}{\textbf{Lifted Heston}} & $10$ & 2.26 & 3.9& 1.20e-03 \\ 
		& {20} & 1.67 & {4.4} & 1.85e-04 \\ 
		& $50$ & 1.3 & 5.2 & 6.81e-05 \\ 
		& $100$ & 1.16 & 6.6  & 2.54e-05\\ 
		& $500$ & 1.04 & 17.4  & 3.66e-06\\ 
		&  &  &  \\ 
		\textcolor{red0}{\textbf{Rough Heston}} & $n \to \infty$ & $r_n \downarrow 1$ & 106.8  & \\ 
		\hline 
	\end{tabular}
	\caption{Convergence of the lifted model towards its rough counterpart for $r_n=1+10\, n^{-0.9}$, with the corresponding computational time in seconds for generating the implied volatility  surface  \eqref{eq:T}-\eqref{eq:K}.}
	\label{tableconvergence} 
\end{table}

All in all, we notice that the number of effective parameters remains constant and does not depend on the number of factors $n$. {This has to be contrasted with the usual multi-factor extensions:  the double Heston model \eqref{eq:doubleS}-\eqref{eq:doubleU}  already has 10 parameters $(U^i_0,\theta_i, \lambda_i , \nu_i, \rho_i )_{i \in \{1,2\}}$,  the multi-scale model of \cite{fouque2011fast} also suffers  from over-parametrization.}\vs2

{In the subsequent subsection, we will explain how  to fix $n$ and $r_n$, so that the parameters to calibrate are reduced to only six effective parameters $(V_0,\theta,\lambda, \nu, \rho, H)$, one additional parameter than the standard Heston model!}

{\subsection{Practical choice of $n$ and $r_n$}}
{We suggest to fix the following values}
\begin{align}\label{eq:choice n r}
n=20 \quad \mbox{ and } \quad r_{20}=2.5.
\end{align}
{Our choice will be based on  the numerical comparison with the rough Heston model of the previous section.} \vs2

We start by explaining our choice for the number of factors $n$ in \eqref{eq:choice n r}. Based on Table \ref{tableconvergence}, we choose $n$ with a good trade-off between time-efficiency and proximity to the rough volatility surface. Fixing $n=20$ seems to be a good choice. Visually, as already shown on Figure \ref{fig:convergence factors},  the two implied volatility slices have almost identical shapes. Whence, one would expect that by letting the parameters $r_{20}$  free,  
    one could  achieve a perfect fit of the rough  surface with only $n=20$ factors. This can be formulated as follows:  keeping the six parameters of the lifted  model fixed as in \eqref{eq: 5parameters},  can one find $  {r^*_{20}(\Theta_0)}>1$ such that 
$$ \sigma_{20}(K,T; r^*_{20}(\Theta_0), \Theta_0 )  \approx \sigma_{\infty}(K, T;\Theta_0), \quad \mbox{for all } K,T?$$

The next subsection provides a positive answer.

\subsubsection{Mimicking roughness by increasing $r_{20}$}

{First, one needs  to understand the influence of the parameter $r_n$ on the \textit{lifted Heston model}.}  Increasing $r_n$ has the effect of boosting the parameters $(c^n_i,x_i^n)_{1 \leq i \leq n}$ in \eqref{eq: ci and xi}, leading to an increase of the  vol-of-vol parameter of the lifted  model  given by  
$\nu \sum_{i=1}^n c_i^n, $
together with faster mean-reversions $(x_i^n)_{1 \leq i \leq n}$ for the factors. In analogy with conventional stochastic volatility models, such as the standard Heston model \eqref{eq:Heston S}-\eqref{eq:Heston V}, increasing the vol-of-vol parameter  together with the speed of mean reversion yields a steeper  skew at the  short-maturity end of the volatility surface. Consequently, increasing the parameter $r_n$ in the lifted model should steepen the implied volatility slice for short-maturities. Figure \ref{fig:convergence 20factors r}  below confirms that this is indeed the case when one increases the value of $r_{20}$ from $1.67$ to $2.8$, for the $20$-dimensional lifted model, as the two slices now almost perfectly match:
\begin{center}
	\includegraphics[scale=0.65]{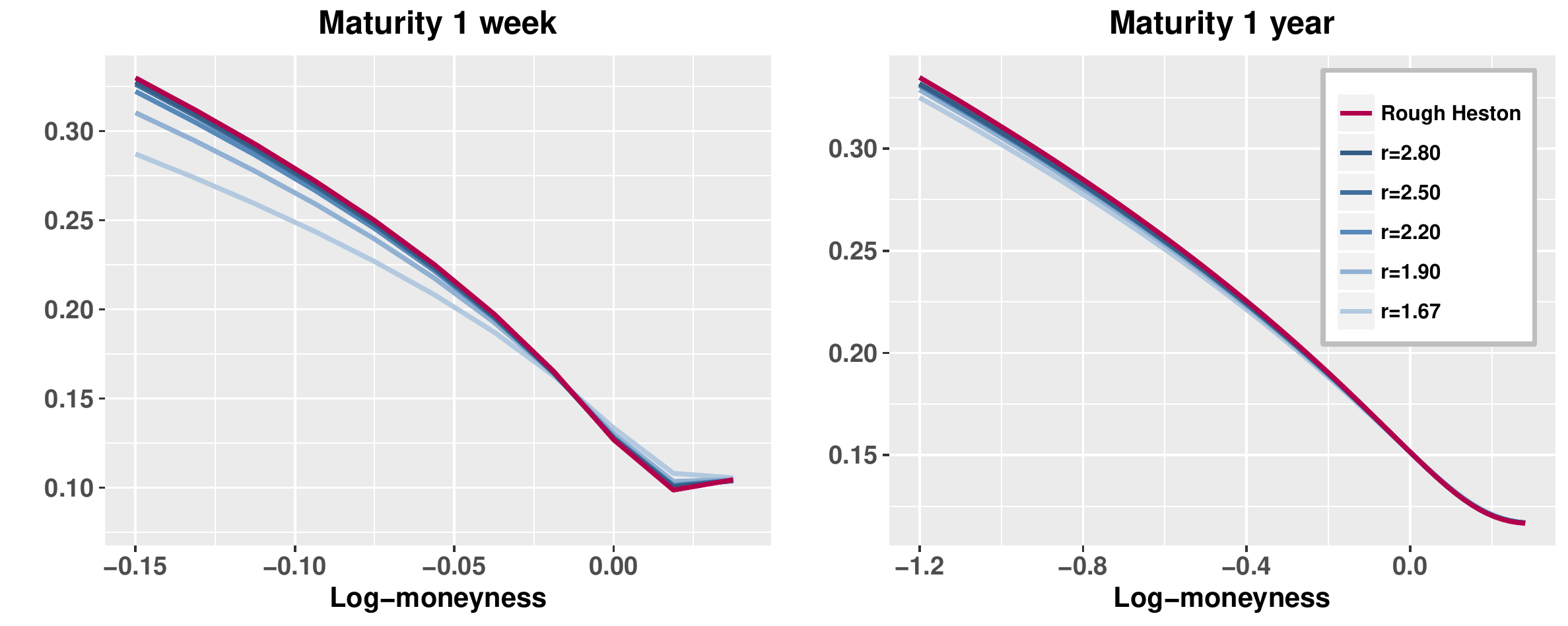}
	\rule{35em}{0.5pt}
	\captionof{figure}{{Implied volatility of the $20$-dimensional lifted model $\sigma_{20}(k,T;r_{20},\Theta_0)$, for different values of $r_{20}$ ranging from $1.67$ to $2.8$, and the rough surface $\sigma_{\infty}(k,T;\Theta_0)$, for two maturities slices $T \in \{1 \mbox{ week}, 1 \mbox{ year}\}$.}}
	\label{fig:convergence 20factors r}
\end{center}

The corresponding mean squared errors of Figure \ref{fig:convergence 20factors r} are collected in Table \ref{table n=20} below.
\begin{table}[h!]
	\centering  
	\begin{tabular}{c c  } 
		\hline
		\textcolor{blue0}{\textbf{Lifted Heston}}  &\textcolor{blue0}{($n=20$)} \\
		\hline
		$r_{20}$ & MSE \\ 
		\hline \hline                 
	  1.67  & 1.85e-04 \\ 
	  1.90  & 4.16e-05 \\ 
    2.20  & 8.72e-06 \\ 
    2.50  & 3.64e-06 \\ 
	2.80  & 2.81e-06 \\ 
		\hline 
	\end{tabular}
	\caption{Mean squared errors between the   $20$-dimensional lifted model $\sigma_{20}(k,T;r_{20},\Theta_0)$ and the rough model $\sigma_{\infty}(k,T;\Theta_0)$, for different values of $r_{20}$. }
	\label{table n=20} 
\end{table}

Because $r_n$ has to converge to $1$, when  $n$ goes to infinity, recall \eqref{eq:rn cond}, we seek to keep $r_n$ as small as possible.  For $n=20$, fixing $r^*_{20}(\Theta_0)=2.5$ yields already satisfactory results, improving the mean squared error  of 1.85e-04 in Table \ref{tableconvergence} to  3.64e-06.\vs2

Before moving to a physical justification of the choice of $r_{20}$,  we proceed to the  full calibration of the \textit{lifted Heston model} with $n=20$ and $r_{20}=2.5$  to the rough volatility surface $\sigma_{\infty}(K, T;\Theta_0)$. That is, we let the    six effective parameters $(V_0,\theta, \lambda,  \nu,  \rho,  H)$ of the lifted model  free. The calibrated values $\hat \Theta_0:=(\hat V_0,\hat \theta,\hat \lambda, \hat \nu, \hat \rho, \hat H)$, provided in Table  \ref{table:calibrate lifted},  agree with \eqref{eq: 5parameters}.  At the visual level,  as shown on Figure \ref{fig:lifted vs rough} in the Appendix, the calibrated lifted surface    is indistinguishable from  the rough surface $\sigma_{\infty}(K,T;\Theta_0)$ for all maturities ranging from one week to two years, with a mean squared error of order 4.01e-07. 
	\begin{table}[H]
		\centering  
		\begin{tabular}{  c c} 
			\hline\hline                        
			Parameters & Calibrated values \\
			\hline \hline                 
			$\hat V_0$&  0.02012504    \\
			$\hat \theta$ & 0.02007956    \\
			$\hat \lambda$ & 0.29300681   \\
			$\hat \nu$ & 0.30527694 \\
			$\hat \rho$ & -0.70241116   \\
			$\hat H$ & 0.09973346\\
			\hline 
		\end{tabular}
		\caption{Calibrated \textit{lifted Heston model} parameters.}
		\label{table:calibrate lifted} 
	\end{table}

We now provide another physical justification for the choice of $r_{20}$ based on the infinite-dimensional  Markovian representation of Remark \ref{R:representation}. We notice that for the {lifted model}, the mean reversions in \eqref{eq: ci and xi} satisfy
$${x^n_i}\geq r^{i-1-n/2}_n, \quad {i=1,\ldots,n }.$$
Therefore, based on  Remark \ref{R:representation}, for $n=20$, one would like to force $x^{20}_{20}$ to be large enough in order to mimic {roughness}  and account for very short timescales, while having $x^{20}_1$ small enough to accommodate a whole palette of timescales. Setting 
  $$ r_{20} \approx  2.5,$$
would cover  mean reversions between $10^{-4}$ and $10^4$.\vs2

\begin{remark}[An alternative way of fixing $r_n$]
{{Fix $n,H$ and maturities $(T_i)_{1 \leq i \leq N}$.   Lemma A.3 in the Appendix suggests to determine the `optimal' value of $r_n^*(H,T_1,\ldots, T_N)$  as  
			$$r_n^*(H,T_1,\ldots, T_N)  = \argmin_{r_n>1} \sum_{i=1}^N w_i \| K_n - K \|_{L^2(0,T_i)}$$
			for some fixed weights $(w_i)_{1 \leq i \leq N}$. For $n=20$, $H=0.1$, $N=1$ and $T=0.1$, $r_{20}^*=2.55$. For $T=1$, $r_{20}^*=3.15$.}}
\end{remark} 

The previous justification suggests that once $n=20$ is fixed, one can choose $r_{20}$ independently  of the parameters  $\Theta$.
The next experiment shows that this is indeed the case.

\subsubsection{Robustness of $r_{20}$: a numerical test}
Throughout this section, we fix  the three parameters $V_0,\theta=0.02$  and $\lambda =0$. In order to verify experimentally  the robustness of $r_{20}=2.5$, we proceed as follows.  
\begin{enumerate}
	\item 
	Simulate $M=500$ set of parameters $(\Theta_k:=(0.02,0.02,0,\nu_k,\rho_k,H_k))_{k=1,\ldots,M}$ uniformly distributed with the following bounds
$$   \quad 0.05\leq \nu\leq 0.5, \quad -0.9 \leq \rho \leq -0.5, \quad 0.05 \leq H \leq 0.2. $$
		\item 
	For each $k=1,\ldots,M$:
	\begin{enumerate}
		\item 
		Generate the rough  volatility surface $\sigma_{\infty}(K,T;\Theta_k)$, for all pairs $(T,K)$  in \eqref{eq:T}-\eqref{eq:K},
		\item 
		Generate the lifted volatility surface $\sigma_{20}(K,T;r_{20}=2.5,\Theta_k)$, for all pairs $(T,K)$  in \eqref{eq:T}-\eqref{eq:K},
		\item 
		Compute the mean squared error between the two volatility surfaces: 
		\begin{align*}
		\mbox{MSE}_k:=\frac{1}{\sum_{(K',T')} w(K',T')} \sum_{(K,T)} w(K,T) (\sigma_{20}(K,T;r_{20}=2.5,\Theta_k) - \sigma_{\infty}(K,T;\Theta_k))^2.
		\end{align*} 
	\end{enumerate} 
\end{enumerate}

The scatter plot and the empirical distribution of the mean squared error $(\mbox{MSE}_k)_{k=1,\ldots,M}$ are illustrated in Figure \ref{fig:mse simulated} below. 
\begin{center}
	\includegraphics[scale=0.65]{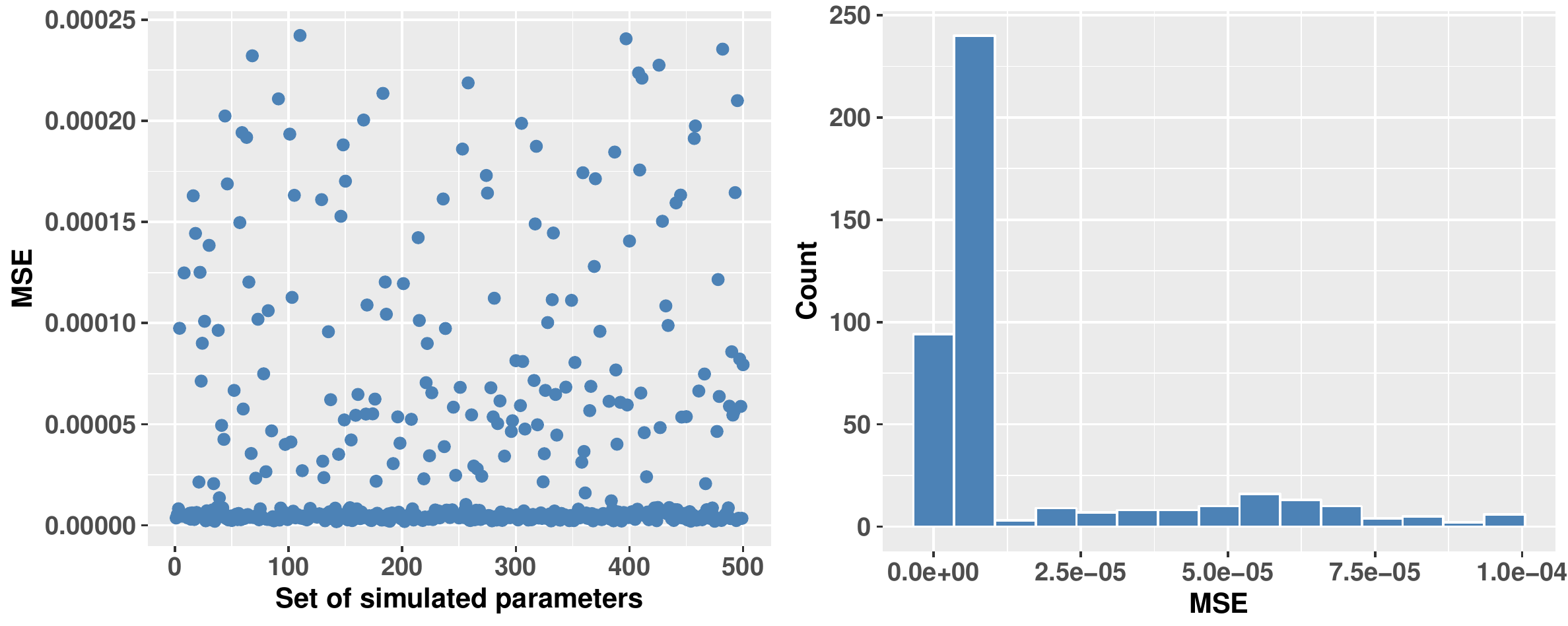}
	\rule{35em}{0.5pt}
	\captionof{figure}{Scatter plot (left) and empirical distribution (right) of the mean squared error $(\mbox{MSE}_k)_{k=1,\ldots,M}$ of the $M=500$  simulated set of parameters $(\Theta_k)_{k=1,\ldots,M}$.}
	\label{fig:mse simulated}
\end{center}

 The first twenty values of the simulated set of parameters with the corresponding mean squared error are provided in Table \ref{table random20} in the Appendix. We observe that the lifted surfaces are quite close  to the rough surface,  for any value of the simulated parameters.
This  is confirmed by  Table \ref{table random} below, where we collect the descriptive statistics of the computed mean squared errors $(\mbox{MSE}_k)_{k=1,\ldots,M}$.

\begin{table}[H]
	\centering  
	\begin{tabular}{cc} 
		\hline\hline                        
		&       MSE \\ 
		\hline
		Minimum &    1.81e-06   \\ 
		1st Quantile   & 3.83e-06   \\ 
		Median  & 5.48e-06   \\ 
		3rd Quantile  & 4.91e-05   \\ 
		Maximum   &  2.42e-04   \\ 
		\hline
	\end{tabular}
	\caption{Descriptive statistics of the mean squared error $(\mbox{MSE}_k)_{k=1,\ldots,M}$ of the $M=500$  simulated set of parameters $(\Theta_k)_{k=1,\ldots,M}$.}
	\label{table random} 
\end{table}


We now show that the mean squared errors can be improved by letting the three parameters $(\nu,\rho,H)$ of the lifted model free. Specifically, consider  the worst mean squared error of  Table \ref{table random}
\begin{align}\label{eq:maxmse}
\max_{\Theta_k} \mbox{MSE}_k = 2.42\mbox{e-}04,
\end{align}
 which is attained for the set of parameters $\Theta_{101}$ with
$$ \nu_{101}= 0.1537099, \quad \rho_{101}=-0.8112745  \; \mbox{ and } \;    H_{101}=0.1892725.$$  Keeping the first three parameters fixed $V_0,\theta=0.02$ and $\lambda=0$, we  proceed to the calibration of the lifted model  to the rough surface $\sigma_{\infty}(K,T;\Theta_{101})$. The calibration yields 
 $$ \hat \nu=0.1647801, \quad \hat \rho= -0.7961080 \; \mbox{ and } \;   \hat H=0.1957235,$$
 improving  the previous mean squared error  \eqref{eq:maxmse} to 1.62e-06. This shows that, by fine tuning the parameters of the lifted model, for any rough volatility surface $\sigma_{\infty}(K,T;  \Theta)$ with a realistic set of parameters $\Theta$,  one can find a set of parameters $\hat \Theta$, not too far from $\Theta$, such that 
$$ \sigma_{20}(K,T;  r_{20}=2.5, \hat \Theta) \approx \sigma_{\infty}(K,T;   \Theta), \quad \mbox{for any pair } (K,T) \mbox{ in } \eqref{eq:T}\mbox{-}\eqref{eq:K}.$$

 To sum up, we showed so far  that  the \textit{lifted Heston model}, with $n=20$ and $r_{20}=2.5$, is able to produce the same volatility surfaces of the rough Heston model, for any realistic set of parameters, for maturities ranging between one week and two years. Consequently, it can be used directly to fit real market data instead of the rough Heston model. 
$$\mbox{\textit{Why is it more convenient to use the lifted Heston model rather than its rough counterpart?}}$$
 On the one hand, it speeds-up calibration time. Indeed, solving numerically the $20$-dimensional system of Riccati ordinary differential equations \eqref{eq:psi xi} is up to twenty times faster than the Adams scheme for the fractional Riccati equation \eqref{eq:RiccatiFrac}.  On the other hand, the lifted model remains Markovian and semimartingale, which opens the door to time-efficient recursive simulation schemes for pricing and hedging more complex exotic options. {Before testing the lifted model in practice, we compare it to the standard Heston model.}\vs2 

\subsection{Comparison with the standard Heston model}
For the sake of comparison, we calibrate a standard Heston model \eqref{eq:Heston S}-\eqref{eq:Heston V} to the full rough volatility surface  $\sigma_{\infty}(K,T;\Theta_0)$, with $\Theta_0$ as in \eqref{eq: 5parameters}. Recall that the standard Heston model corresponds to the case $n=1$, $x^1_1=0$ and $c^1_1=1$. The calibrated  parameters of the standard Heston are provided in Table \ref{table:calibrate Heston} below.  We observe that the calibrated values of $(\hat V_0,\hat \theta, \hat \rho)$ have the same magnitude as the ones of  \eqref{eq: 5parameters}. This is not surprising since these parameters have the same interpretation in the two models: the first two parameters $(\hat V_0,\hat \theta)$ govern the level of the term structure of forward variance at time $0$ while $\rho$ dictates the leverage effect between the stock price and its variance.  
\begin{table}[H]
	\centering  
	\begin{tabular}{  c c} 
		\hline\hline                        
		Parameters & Calibrated values \\
		\hline \hline                 
		$\hat V_0$&  0.019841  \\
		$\hat \theta$ & 0.032471  \\
		$\hat \lambda$ & 3.480784  \\
		$\hat \nu$ & 0.908037\\
		$\hat \rho$ & -0.710067 \\
		\hline 
	\end{tabular}
	\caption{Calibrated Heston model parameters.}
	\label{table:calibrate Heston} 
\end{table}

Despite the extreme values of the  calibrated mean reversion and vol-of-vol parameters $(\hat \lambda, \hat \nu)$, {the Heston model is not able to reproduce the steepness of the skew} for short maturities  as shown on 
Figure \ref{fig:calibrated heston surface} in the Appendix, with a mean squared error of order 2.06e{-03}. For long maturities, the fit is fairly good. \vs2

In order to compare our findings  with the observed  stylized fact of Figure \ref{fig:skew intro}, we plot on Figure \ref{fig:skew} below the term structure of the at-the-money skew of the three models: the rough Heston with parameters as in \eqref{eq: 5parameters}, the calibrated \textit{lifted Heston} model of Table \ref{table:calibrate lifted} and the calibrated Heston model of Table \ref{table:calibrate Heston}.  The Heston model fails in reproducing the explosive behavior of the term structure of the at-the-money skew observed in the market. On the contrary, this feature is captured by the lifted and rough counterparts.   For long maturities, all three model have the same behavior.  
\begin{center}
	\includegraphics[scale=0.6]{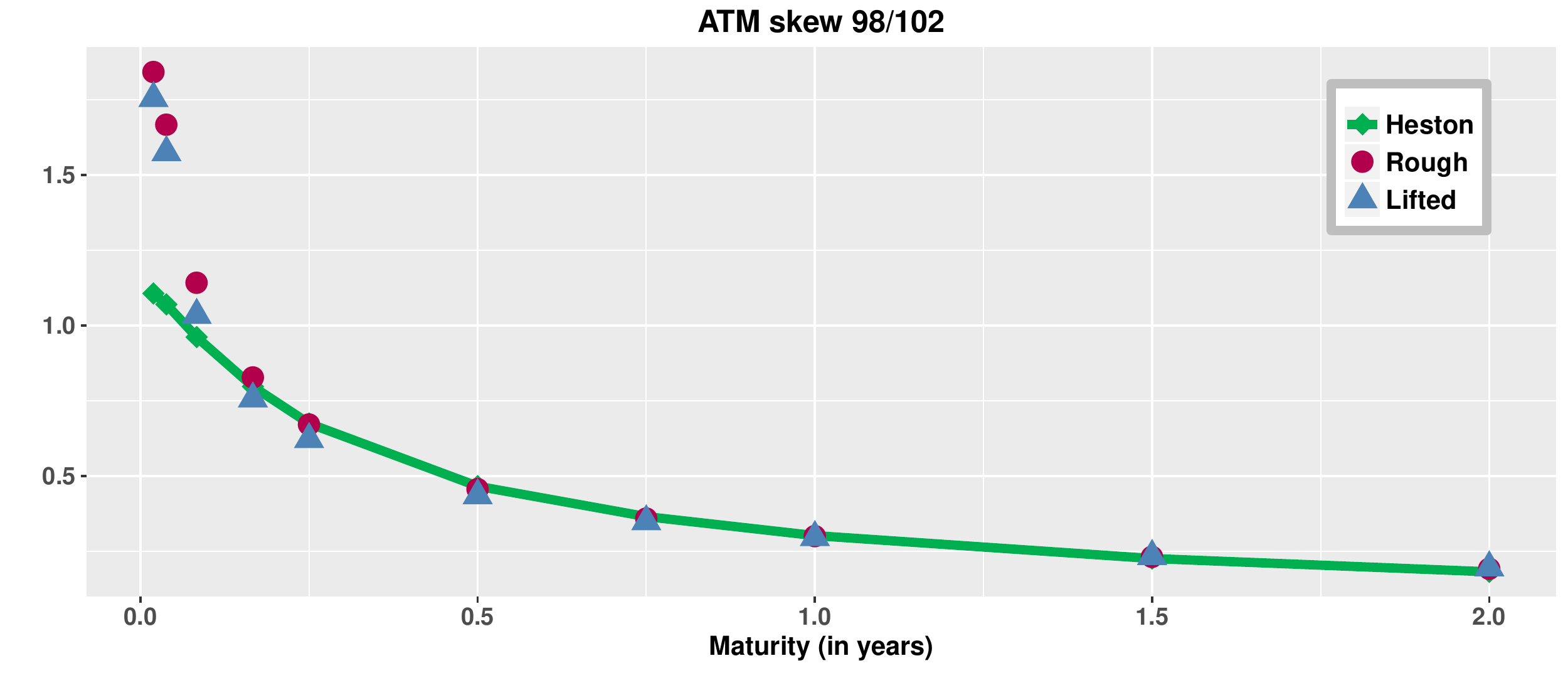}
	\rule{35em}{0.5pt}
	\captionof{figure}{{Term structure of the at-the-money skew of the rough Heston model $\sigma_{\infty}(K,T;\Theta_0)$ of \eqref{eq: 5parameters} (red circles), the calibrated \textit{lifted Heston model} $\sigma_{20}(K, T; r_{20}=2.5, \hat \Theta_0)$ of Table \ref{table:calibrate lifted} (blue triangles) and the calibrated Heston model of Table \ref{table:calibrate Heston} (green line). }}
	\label{fig:skew}
\end{center}

In the sequel, we will  show that, for $n=20$ factors, the \textit{lifted Heston model} provides  an appealing  trade-off between consistency with market data and tractability. We stress that $r_{20}=2.5$ is  kept fixed in the lifted model, which now has only six effective parameters to calibrate $(V_0,\theta,\lambda,\nu,\rho,H)$. Again, in practice, $V_0$ and $\Theta_0$ can be eliminated by specifying the initial forward variance curve as input and $\lambda$ can be set to $0$, as mean reversions at different speeds are naturally encoded in the lifted model through the family $(x^n_i)_{1 \leq i \leq n}$. By doing so,  one reduces the effective number of parameters to only three  $(\nu,\rho,H)$, as already done in \cite{euch2017roughening} for the rough Heston model.

\section{Calibration on market data and simulation}\label{S:numerics}
In this section, we fix the number of factors to $n=20$ and set $r_{20}=2.5$ in \eqref{eq: ci and xi}.  We  demonstrate	 that the \textit{lifted Heston model}: 
\begin{itemize}
\item
{captures the explosion of the  at-the-money skew observed in the market,}
\item
is easier to simulate than the rough model,
\item
tricks the human eye as well as the statistical estimator of the Hurst index. 
\end{itemize}

{\subsection{Calibration to the at-the-money skew}}
Going back to real market data, we calibrate the lifted model to the at-the-money skew of Figure \ref{fig:skew intro}. Keeping the parameters
$V_0=0.02$, $\theta=0.02$ and $\lambda=0$ fixed, the calibrated parameters are given by 
\begin{align}\label{eq:calibrated skew lifted}
 \hat \nu = 0.3161844, \quad \hat \rho=-0.6852625 \quad \mbox{and} \quad \hat H= 0.1104290.
\end{align} 
The fit is illustrated on Figure \ref{fig:skew calibrated} below.

\begin{center}
	\includegraphics[scale=0.6]{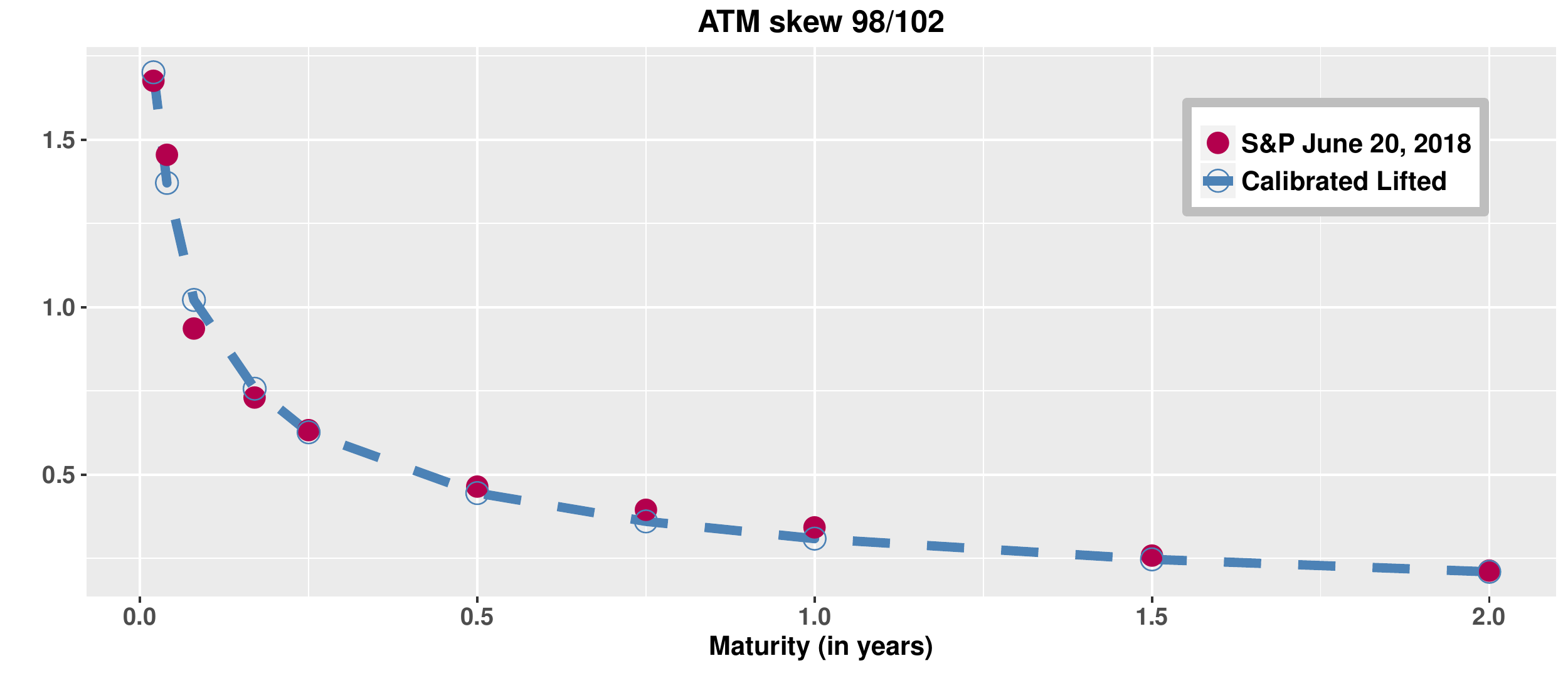}
	\rule{35em}{0.5pt}
	\captionof{figure}{{Term structure of the at-the-money skew for the S\&P index on June 20, 2018 (red dots) and for the lifted model with calibrated parameters \eqref{eq:calibrated skew lifted} (blue circles with dashed line).} }
	\label{fig:skew calibrated}
\end{center}

We notice the calibrated value $\hat H$ in \eqref{eq:calibrated skew lifted} is coherent with the value $(0.5-0.41)=0.09$, which can  be read off  the power-law fit of Figure \ref{fig:skew intro}.   Consequently, in the pricing world, the parameter $H$ quantifies the explosion of the at-the-money skew through a power-law $t\to C t^{0.5-H}$, see also \cite{fukasawa2017short}.\vs2

 {We discuss the simulation procedure  and the statistical estimation of $H$ of our lifted model in the next subsection.}

\subsection{Simulation and estimated roughness}

{Until now, there is  no existing scheme to simulate the variance process \eqref{eq:rough variance} of the rough Heston model, the crux resides in the non-Markovianity of the variance process,  the singularity of the  kernel and the square-root dynamics. In contrast, numerous approximation schemes have been developed for the simulation of the standard square-root process \eqref{eq:Heston V}, see \cite[Chapters 3 and 4]{Al15} and the references therein.} Because the \textit{lifted Heston model}	 \eqref{eq:liftedS}-\eqref{eq:liftedU} is a Markovian and semimartingale model, one can adapt standard recursive Euler-Maruyama schemes to simulate the variance process $V^n$ first, and then the stock price $S^n$.  For $T>0$, we consider the  modified explicit-implicit scheme \eqref{eq:simliftedV}-\eqref{eq:simliftedU} detailed in the Appendix for the variance process $V^n$. \vs2

We observe on Figure \ref{fig:simfactors}  below that the factors $(U^{20,i})_{1 \leq i \leq 20}$ are highly correlated. We can distinguish between the  short-term factors with fast mean reversions, responsible of the `roughness', and the long-term factors, with slower mean reversions, determining the level of the variance process. The variance process is then obtained by aggregating these factors with respect to \eqref{eq:liftedV}. {We also notice that some of the factors $(U^{n,i})_{1 \leq i \leq n}$  become negative,  but that the aggregated process $V^n$ remains nonnegative at all time.}

{\begin{remark}[Nonnegativity of the variance process]\label{R:nonnegativity} Looking at the stochastic differential equation \eqref{eq:liftedV}-\eqref{eq:liftedU}, it is not straightforward at all why $V^n$ should stay nonnegative at all time, even for the zero initial curve $g_0 \equiv 0$. Indeed, some of the factors $(U^{n,i})_{1 \leq i \leq n}$ may become negative, but surprisingly enough, their aggregated sum $V^n$ remains nonnegative, at all time.  This is due to a very special underlying structure: equations \eqref{eq:liftedV}-\eqref{eq:liftedU} can be recast as a stochastic Volterra equation of convolution type for a suitable kernel, we refer to Appendix \ref{eq:appendix existence} for more details.
\end{remark}}

\begin{center}
	\includegraphics[scale=0.6]{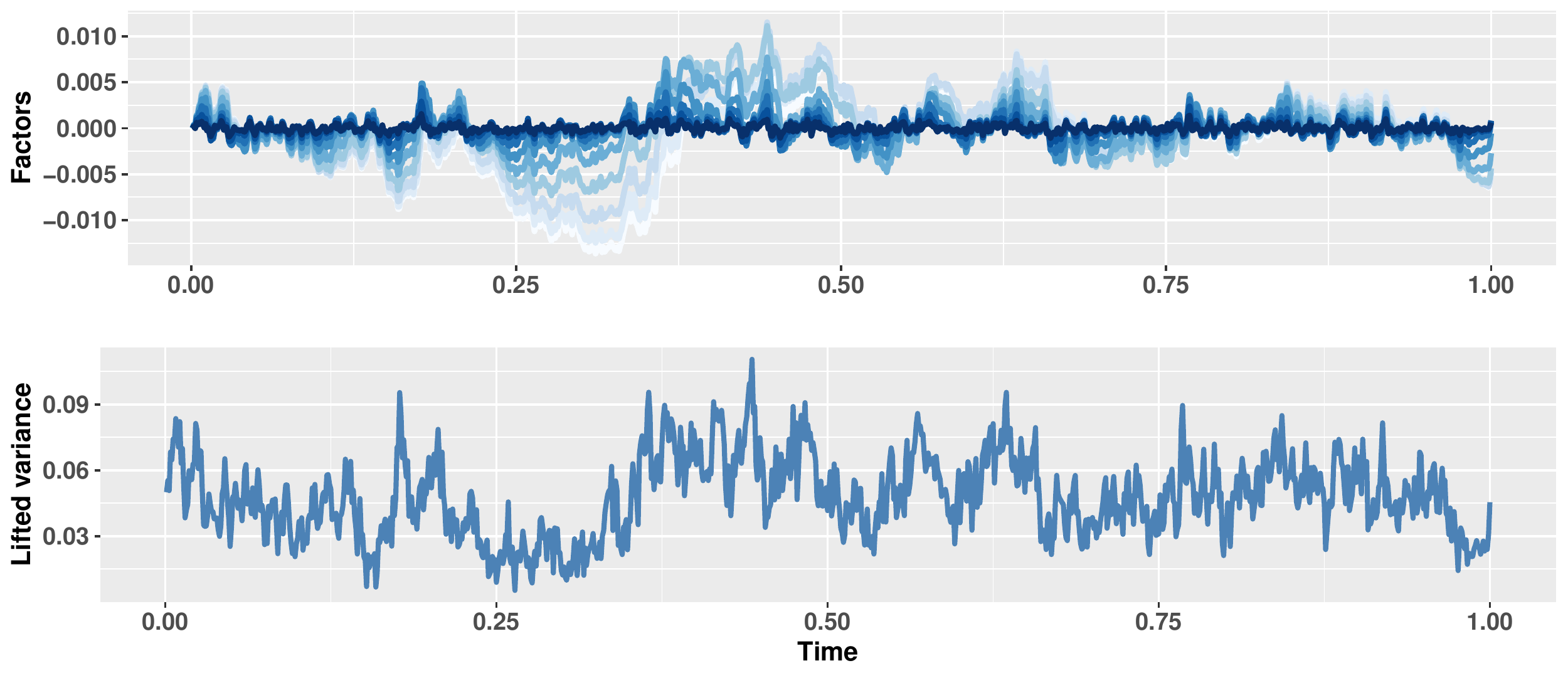}
	\rule{35em}{0.5pt}
	\captionof{figure}{{One sample path of the simulated factors $(U^{20,i})_{1 \leq i \leq 20}$ with blue intensity proportional to the speed of mean reversions $(x^n_i)_{1 \leq i \leq 20}$ (upper) and the corresponding aggregated variance process $V^n$ (lower) with  parameters $V_0=0.05$, $\theta=0.05$, $\lambda=0.3$, $\nu=0.1$  and $H=0.1$ for a time step of $0.001$ and $T=1$ year.}}
	\label{fig:simfactors}
\end{center}

Visually, the sample path of the variance process seems rougher than the one of a standard Brownian motion, closer to that of a fractional Brownian motion with a small Hurst index. This observation is strengthened by     Figure \ref{fig:sim comparison} below: the sample path of the volatility process in the   lifted  model (lower subgraph) looks clearly rougher than the  sample path of the volatility process in the standard Heston model (middle subgraph). It also seems to enjoy the same regularity as that of 
 the realized volatility of the  S\&P (upper subgraph). \vs2

In what follows, we provide a  quantitative analysis of the previous observation by running two standard statistical experiments that have been used in  \cite{BLP:16,gatheral2014volatility} to estimate the roughness of a realized volatility time series. {More precisely, 
	the empirical studies of \cite{BLP:16,gatheral2014volatility} on  a very wide range of assets volatility time series    revealed that the dynamics of  the log-realized volatility are close to that of a fractional Brownian motion with a `universal' Hurst parameter $H$ of order $0.1$, from intra-day up to daily timescales. {These studies  provide a  physical interpretation of   the parameter $H$, as it measures the roughness of the empirical realized volatility of the upper graph of Figure \ref{fig:sim comparison}.}} We run these two procedures on a simulated path of the lifted model.
First, we apply the estimation procedure of  \cite{gatheral2014volatility}  for daily timescales. Then, we  apply the methodology that was used in \cite{BLP:16}, focusing on intra-day timescales. We recall that, theoretically speaking, because the lifted variance process is a semimartingale, it has the same regularity as a standard Brownian motion, that is $H=0.5$.

\begin{center}
	\includegraphics[scale=0.6]{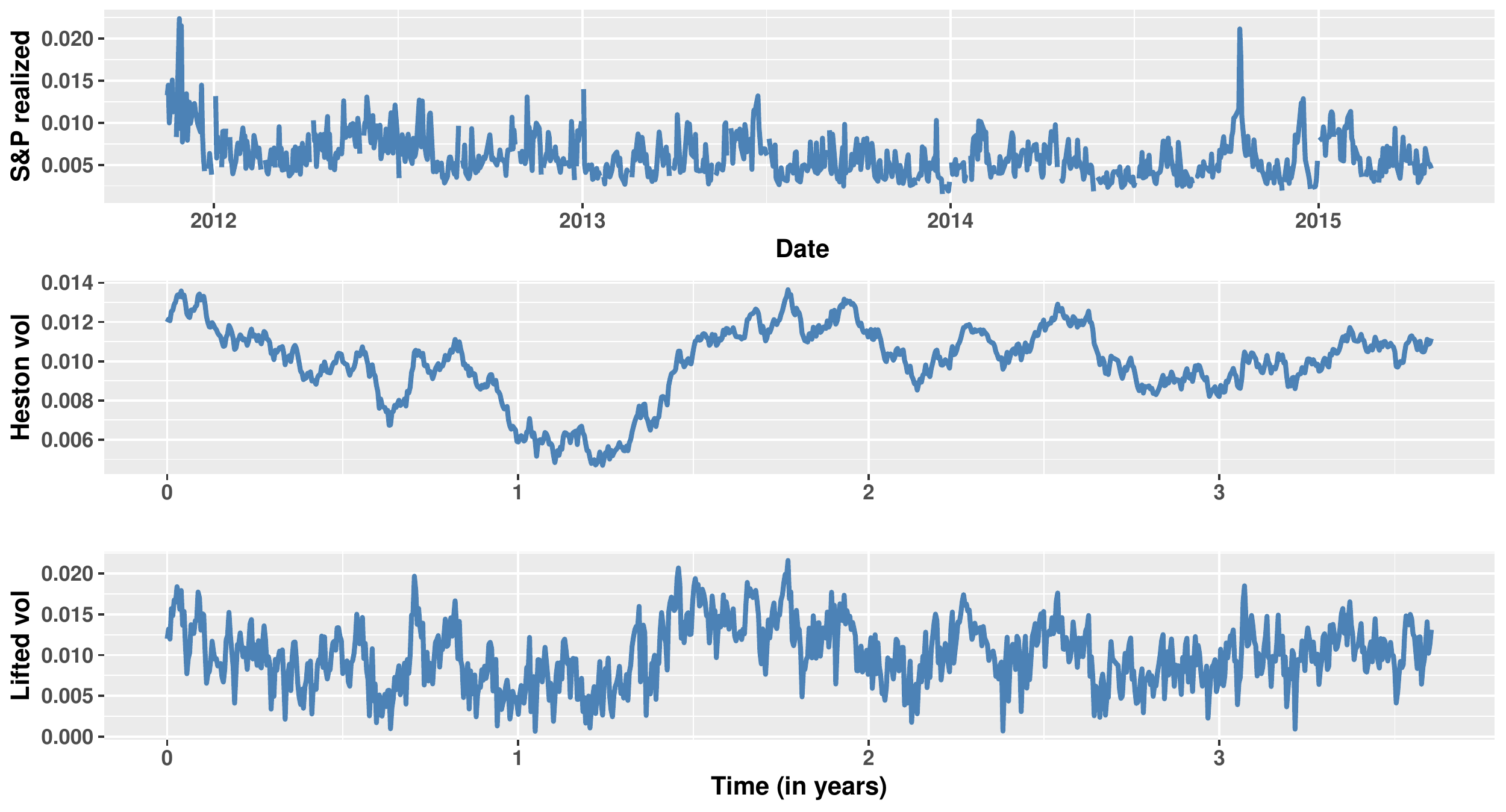}
	\rule{35em}{0.5pt}
	\captionof{figure}{{Estimated Hurst index of:  the realized volatility of the  S\&P\textsuperscript{(a)} (upper),  a sample path of the volatility process in the Heston model (middle), and  a sample path of the volatility process in the  the lifted  model with $H=0.1$ (lower). The simulation is run with $N=250$ time steps for each year.}}
	\footnotesize\textsuperscript{(a)}The realized volatility data series can be downloaded from \url{https://realized.oxford-man.ox.ac.uk/}.
	\label{fig:sim comparison}
\end{center}

\subsubsection{First statistical experiment for daily timescales as in \cite{gatheral2014volatility}}
We replicate the methodology used in \cite[Section 2]{gatheral2014volatility} for estimating the smoothness of the volatility process $\sigma=\sqrt{V^n}$.\footnote{More details can be found in the Python notebook of Jim Gatheral  \url{https://tpq.io/p/rough_volatility_with_python.html}.} This boils down to  estimating the following $q$-variation 
\begin{align}\label{eq:qvariation}
\hat m(q,\Delta) = \frac{1}{N}  \sum_{k=1}^N \left|\sigma_{k\Delta} - \sigma_{(k-1)\Delta}\right|^q.
\end{align} 
for different values of $q$ and timescales $\Delta$ greater than one day. We recall that the notion of $q$-variation is linked to the notion of Besov smoothness of stochastic processes, and to that of H\"older regularity as $\Delta \to 0$, see \cite{rosenbaum2009first}. In order to estimate \eqref{eq:qvariation}, we simulate a sample path of the lifted variance model $V^{n}$ (recall that $n=20$) with $H=0.1$ and $T=2$ years with $N=500$ time steps. This corresponds to one time step per day with  the convention of $250$ trading days per year. On the left hand side of Figure \ref{fig:estimateHexperiment1} below, we plot the value of  $ \log \hat m(q,\Delta)$ against $\log \Delta$, for  $\Delta=1,2,\ldots,100$ days and $q \in \{0.5,1,1.5,2,3\}$.

 \begin{center}
 	\includegraphics[scale=0.65]{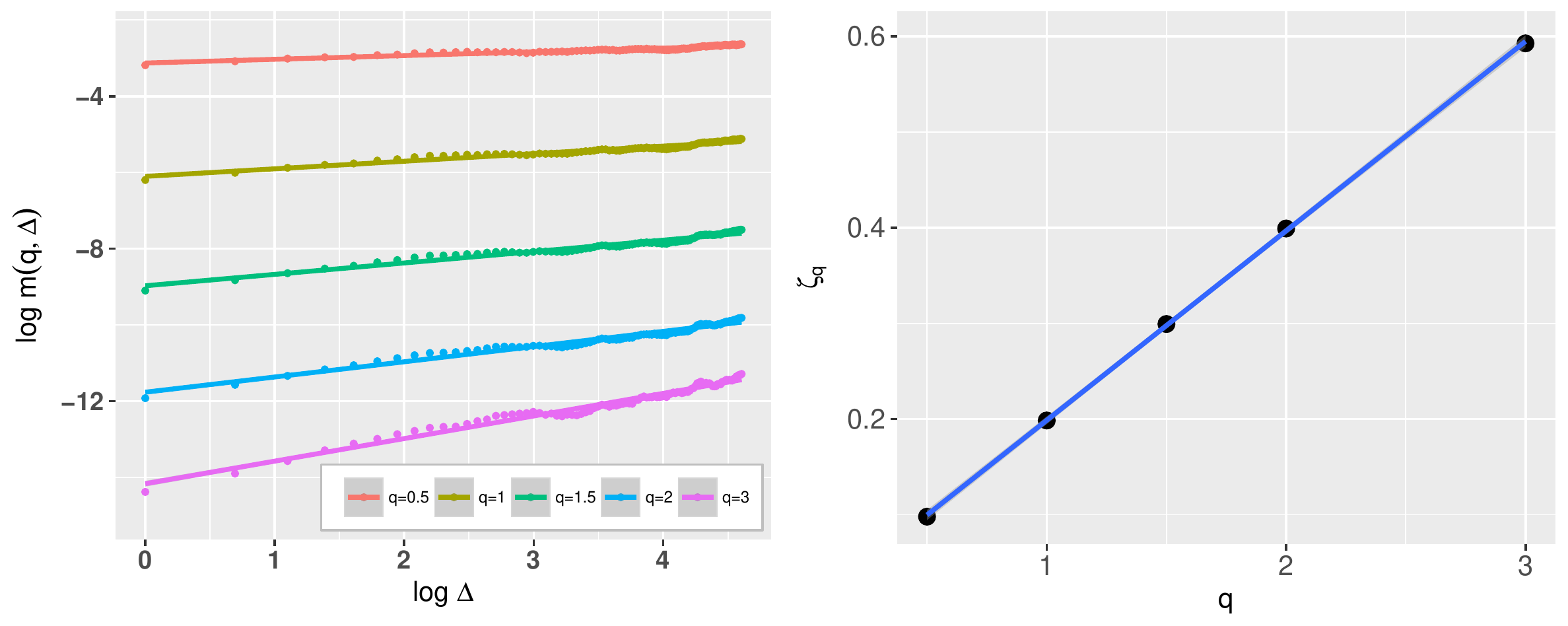}
 	\rule{35em}{0.5pt}
 	\captionof{figure}{Estimation procedure of \cite{gatheral2014volatility} applied to a sample path of the lifted volatility process: log-log plot for different estimated moments $\hat m(q,\Delta)$ (left); $\zeta_q $ against $q$ (right). }
 	\label{fig:estimateHexperiment1}
 \end{center}

  For each $q$, the points seem to lie on a straight line, which suggests the following scaling
 $$  \hat m(q,\Delta) = K_q \Delta^{\zeta_q}, $$ 
 where $\zeta_q>0$ corresponds to the slope of the fitted line in the log-log plot. Further, plotting  $\zeta_q$ against $q$ on the right hand side of Figure  \ref{fig:estimateHexperiment1}, shows that 
 $$ \zeta_q \approx \hat H q, \quad \hat H=0.19.$$
 To sum up, this shows that, statistically speaking, the estimated $q$-variation  of the lifted volatility process enjoys the  following scaling  
 $$ \hat m(q,\Delta) = K_q \Delta^{q \hat H}, $$
 similar to that of a fractional Brownian motion with Hurst index $H=0.19$. Consequently, at daily timescales, the simulated  volatility process of the \textit{lifted Heston model}  not only tricks  the human eye, but also misleads the  estimator of the Hurst index used in \cite{gatheral2014volatility} with an estimated $\hat H=0.19$, way below $0.5$. What about intra-day timescales? We provide an answer in the following subsection.

 \subsubsection{Second statistical experiment for intra-day timescales as in \cite{BLP:16}}
 In  \cite{BLP:16}, an efficient estimator for $H$ based on the autocorrelation function is applied for intra-day timescales, ranging from couple minutes to a day. The estimated $H$ is determined by the following linear regression 
 \begin{align}\label{eq:reg}
 \log(	1 - \rho_{\sigma}(k\Delta))  = b + 2H \log (k\Delta), \quad k=1,\ldots,K,
 \end{align} 
 where $\rho_{\sigma}$ is the autocorrelation function of the time series of the volatility $\sigma$. We refer to \cite[section 2.3.1]{BLP:16} for more details on the estimation procedure. \vs2
 
 We apply  the same methodology that was used in  \cite{BLP:16} to the lifted volatility process $\sigma=\sqrt{V^n}$ with $H=0.1$ and $n=20$. We set $T=2$ years and we simulate one sample path of the lifted volatility process $\sigma$   for $N=3 \times 10^{5}$ times steps. To fix ideas, in a high frequency trading environment, this corresponds roughly to one time step per  minute, under  the convention of $250$ trading days per year and $10$ hours per day.  Then,  from this simulated sample path, we extract subsamples of length $500$ with different time steps $\Delta \in \{1 \mbox{min}, 5 \mbox{min}, 10 \mbox{min}, 30 \mbox{min}, 1 \mbox{h}, 2 \mbox{h}, \ldots, 1 \mbox{day}  \}$, and we estimate  $H$ on each  subsample using the regression \eqref{eq:reg}. The estimated values,  illustrated on  Figure \ref{fig:estimatedH} below,  are aligned with the previous estimation $\hat H = 0.19$, for any timescale greater than $10$ minutes. The estimator converges towards the true value $0.5$ only for very short timescales which are less than 10 minutes.  
 One can compare Figure \ref{fig:estimatedH} with \cite[Figure 3]{BLP:16}.

 \begin{center}	
 	\includegraphics[scale=0.55]{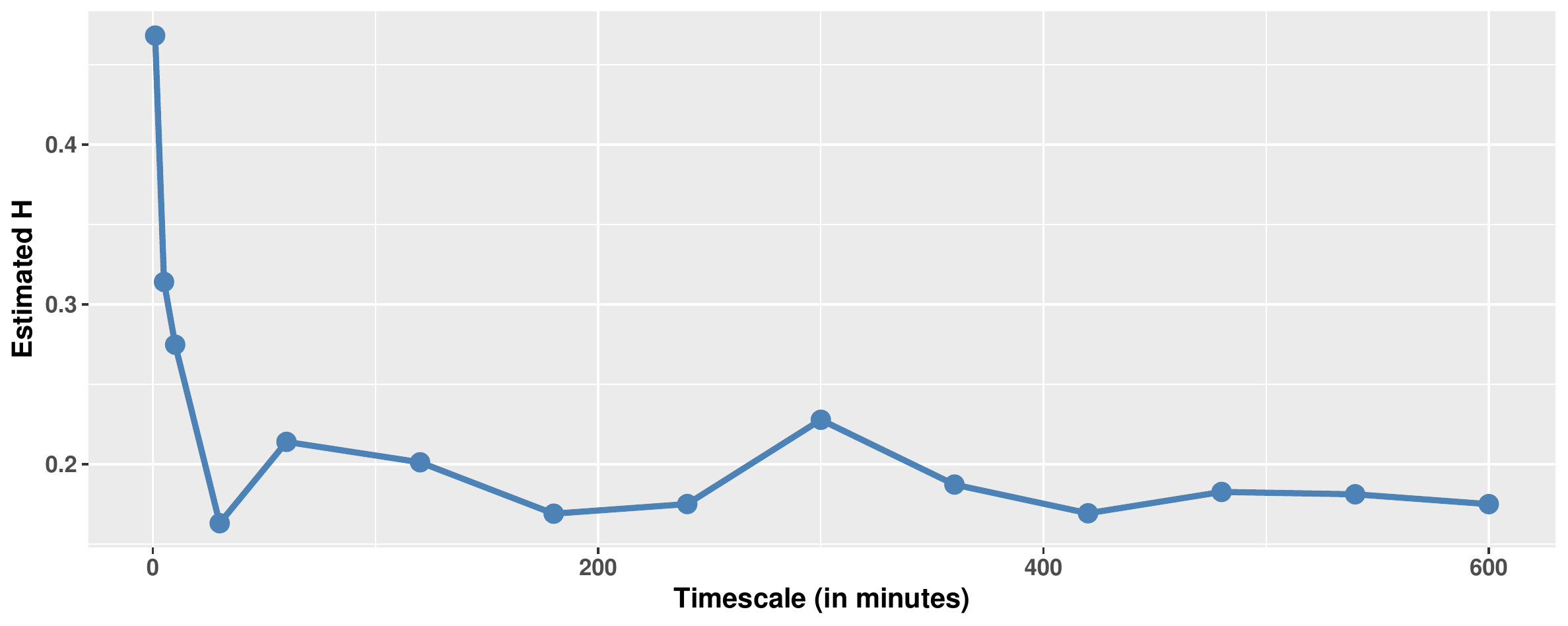}
 	\rule{35em}{0.5pt}
 	\vspace{-0.3cm}
 	\captionof{figure}{{Estimated Hurst index of the simulated path of the lifted volatility process  with $H=0.1$ for different timescales ranging from 1 minute up to 1 day with $n=20$ factors.}} 
 	\label{fig:estimatedH}
 \end{center}

 We point out that the estimator recognizes a semimartingale model for the simulated volatility of the Heston model, with an estimated $H$
close to $0.5$ and displays a value of $\hat H=0.11$ for the S\&P.
The {lifted model} is therefore capable of mimicking,  up to some extent, the `roughness' of the volatility observed on the market, even for short intra-day timescales. This should be paralleled with the explosive-like behavior of the at-the-money skew encountered earlier on Figures \ref{fig:skew}-\ref{fig:skew calibrated}. {Stated otherwise, if one is only provided the lower graph of Figure \ref{fig:sim comparison}, one cannot conclude whether the path has been generated by a rough volatility model with Hurst index $H=0.19$ or by our lifted model with $H=0.1$, for any reasonable timescale.  {As the timescale goes to $0$, the estimated value for $H$ of the lifted model has to converge to $0.5$, since  $V^n$ is a semimartingale, and therefore has the same regularity as a standard Brownian  motion. However, depending on the number of factors, finer timescales are needed for  the estimator  to recognize a semimartingale model with an estimated $\hat H$ close to $0.5$. This is illustrated on Figure \ref{fig:estimatedH50} below, where the same experiment is carried for $n=50$ factors and $r_n=1.8$. }}

  	\vspace{-0.1cm}
 
 \begin{center}	
 	\includegraphics[scale=0.55]{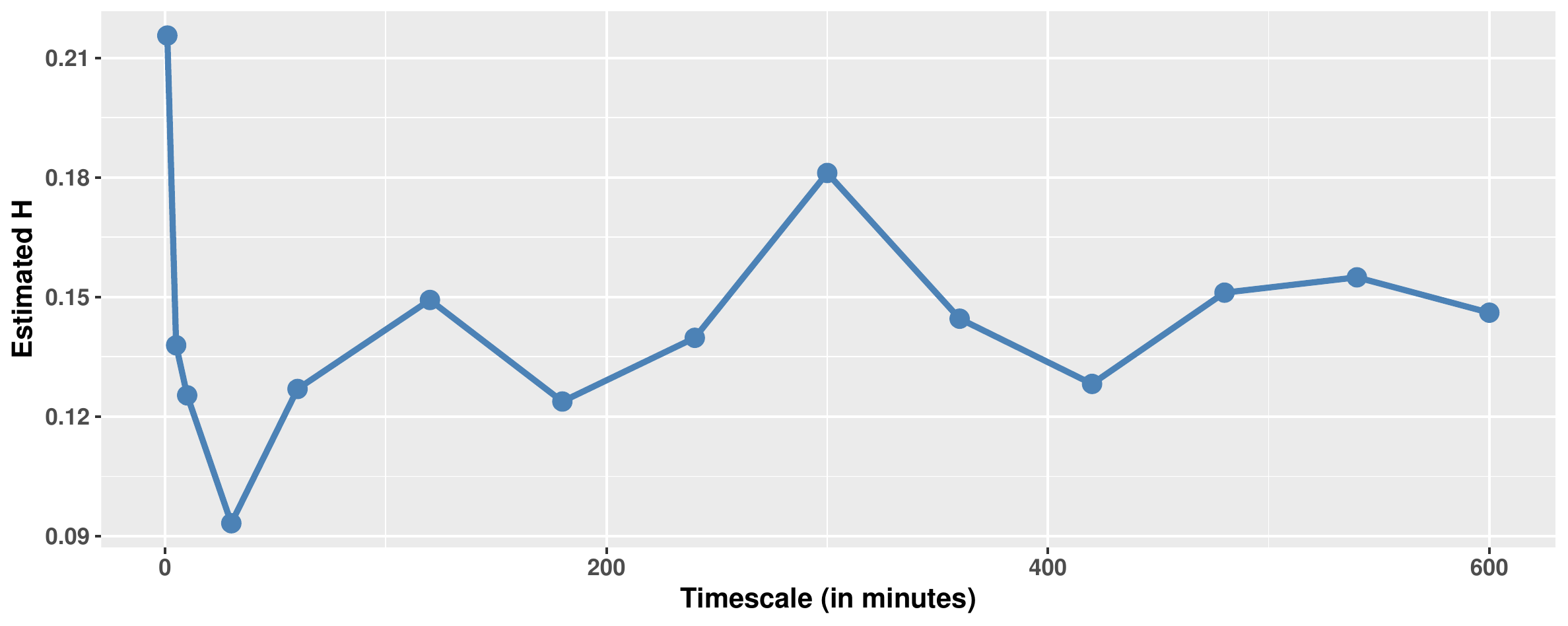}
 	\rule{35em}{0.5pt}
 	\vspace{-0.3cm}
 	\captionof{figure}{{Estimated Hurst index of the simulated path of the lifted volatility process with $H=0.1$ for different timescales ranging from 1 minute up to 1 day with $n=50$ factors.}} 
 	\label{fig:estimatedH50}
 \end{center}

\section{Conclusion}

We introduced  \textit{the lifted Heston model}, a conventional multi-factor stochastic volatility model, where the factors share the same one-dimensional Brownian motion but mean revert at different speeds corresponding to different timescales. The model nests as extreme cases the standard   {Heston model} (for $n=1$ factor), and the rough Heston model (when $n$ goes to infinity). Inspired by rough volatility models, we provided a good parametrization of the  model reducing the number of parameters to calibrate: the model  has only one additional effective parameter  than the standard Heston model, independently of the number of factors. The first five parameters have the same interpretation as in the standard Heston model, whereas the additional one  has a physical interpretation as it is linked to the regularity of the sample paths and the explosion of the at-the-money skew.\vs2

This sheds some new light on the reason behind the remarkable fits of rough volatility models. Indeed, a rough variance process can be seen as a superposition of infinitely many  factors sharing the same one-dimensional Brownian motion but mean reverting at different speeds ranging from $0^+$ to $\infty$. Each factor corresponds to a certain timescale.  Therefore, time multi-scaling is naturally encoded in rough volatility models, which explains why these models are able to jointly handle different maturities in a satisfactory fashion.\footnote{Multiple timescales in the volatility process have been identified  in the literature, see for instance \cite[Section 3.4]{fouque2011multiscale}.} \vs2

{Finally, Table \ref{tablecomparisonmodels conclusion} below compares the characteristics of the three different models. As it can be seen, the \textit{lifted Heston model} possesses an appealing trade-off between flexibility and tractability!}

	\begin{table}[h!]
	\centering  
	\begin{tabular}{c c  c c} 
		\hline\hline                        
		&  & Stochastic volatility models  &    \\
		\hline
		
		Characteristics & Heston & Rough Heston & Lifted Heston \\
		\hline \hline                 
		Markovian & \textcolor{blue0}{\cmark} & \textcolor{red0}{\xmark}  & \textcolor{blue0}{\cmark} \\
		Semimartingale & \textcolor{blue0}{\cmark} & \textcolor{red0}{\xmark}  & \textcolor{blue0}{\cmark} \\
		Simulation & \textcolor{blue0}{Fast} & \textcolor{red0}{Slow}  & \textcolor{blue0}{Fast} \\
		&&&
		\\
		{Affine Volterra process} & \textcolor{blue0}{\cmark} & \textcolor{blue0}{\cmark} & \textcolor{blue0}{\cmark} \\
		Characteristic function & \textcolor{blue0}{Closed} & \textcolor{red0}{Fractional Riccati}  & n-Riccati \\
		Calibration  & \textcolor{blue0}{Fast} & \textcolor{red0}{Slower} & \textcolor{blue0}{20x rough}\textsuperscript{(a)} \\
		&&&\\
		Fit short maturities & \textcolor{red0}{\xmark} & \textcolor{blue0}{\cmark} &\textcolor{blue0}{\cmark} \\
		Estimated daily regularity & \textcolor{red0}{$H \approx 0.5$} & \textcolor{blue0}{$H \approx 0.1$} & $H \approx 0.2$ \\
		\hline 
	\end{tabular}
	\caption{Summary of the characteristics of the different models. {\footnotesize\textsuperscript{(a)}for $n=20$. }}
	\label{tablecomparisonmodels conclusion} 
\end{table}



\begin{appendices}

\section{Existence, uniqueness and rough limiting model}\label{eq:appendix existence}
{
	In the sequel, the symbol $*$ stands for the convolution operation, that is $(f*\mu)(t)= \int_0^t f(t-s) \mu(ds)$ for any suitable function $f$ and measure $\mu$. For a right-continuous  function $f$  of locally bounded variation, we denote by $df$ the measure induced by its distributional derivative, that is 
	$f(t) = f(0) + \int_{(0,t]} df(s)$.}
	  
	  \subsection{Existence and uniqueness}
	  {
	We provide in this section the strong existence and uniqueness of \eqref{eq:liftedS}-\eqref{eq:liftedU}, for a fixed $n \in \mathbb N$. We start by noticing that  \eqref{eq:liftedS} is equivalent to 
$$ S^n_t = \mathcal E\left(\int_0^t V^n_s dB_s\right), \quad t \geq 0,$$
where $\mathcal E$ is the Doléans-Dade exponential. Therefore, it suffices to prove the existence and uniqueness of \eqref{eq:liftedV}-\eqref{eq:liftedU}. Formally, starting from a solution to \eqref{eq:liftedV}-\eqref{eq:liftedU}, the variation of constants formula on \eqref{eq:liftedU} yields
\begin{align}\label{eq: volterra Un}
U^{n,i}_t = \int_0^t e^{-x^n_i (t-s)} \left(-\lambda V^n_s ds+ \nu \sqrt{V^n_s} dW_s \right),  \quad i=1,\ldots,n, 
\end{align}
so that \eqref{eq:liftedV} reads
\begin{align}\label{eq:volterra Vn}
V^{n}_t = g^n_0(t) + \int_0^t K^n(t-s) \left(-\lambda V^n_s ds+ \nu \sqrt{V^n_s} dW_s \right),   
\end{align}
where $K^n$ is the following completely monotone\footnote{{A function $f$ is said to be completely monotone, if it is infinitely differentiable on $(0,\infty)$ such that $(-1)^p f^{(p)} \geq 0$, for all $p \in \mathbb N$.}} kernel
\begin{align}\label{eq:kernel Kn}
K^n(t) = \sum_{i=1}^n c^n_i e^{- x^n_i t }, \quad  t\geq 0.
\end{align} 
Whence, if one proves the uniqueness of \eqref{eq:volterra Vn}, then, uniqueness of \eqref{eq:liftedU} follows by virtue of \eqref{eq: volterra Un}. Conversely, if one proves the existence of a nonnegative solution $V^n$ to \eqref{eq:volterra Vn}, then, one can define  $(U^{n,i})_{1 \leq i \leq n}$ as in \eqref{eq: volterra Un}, showing that $(V^n,(U^{n,i})_{1 \leq i \leq n})$ is a solution to \eqref{eq:liftedV}-\eqref{eq:liftedU}. Therefore, the problem is reduced to proving the existence and uniqueness for the stochastic Volterra equation \eqref{eq:volterra Vn}.} \vs2

{In \cite{AJEE18b}, the existence of a nonnegative solution  to \eqref{eq:volterra Vn} is proved, provided the initial input curve $g^n_0$ satisfies a  certain `monotonicity' condition. This condition is related to the resolvent of the first kind $L^n$ of the  kernel \eqref{eq:kernel Kn}, which is defined as the unique measure satisfying  
	\begin{equation*}
	\int_0^t K^n(t-s)L^n(ds)=1, \quad t \geq 0.\footnote{{The existence of $L^n$ is ensured by the complete monoticity of $K^n$, see \cite[Theorem 5.5.4]{GLS:90}.}}
	\end{equation*}
	 More precisely, denoting by $\Delta_h$ the semigroup of right shifts acting on continuous functions, i.e.~$\Delta_h f=f(h+ \cdot)$ for $h\geq 0$, $g_0^n$ should satisfy
	\begin{equation} \label{Croissance}
	\Delta_h g^n_0 - (\Delta_hK^n * L^n)(0) g^n_0 - d(\Delta_hK^n * L^n) * g^n_0 \geq 0, \quad  h \geq 0, \footnote{
		{One can show that   $\Delta_hK^n*L^n$  is right-continuous and of locally bounded variation, thus the associated measure $d(\Delta_hK^n*L^n)$ is well defined.}}
	\end{equation}
	leading to the following definition of the set $\Gc_{K^n}$ of {admissible input curves}:
	\begin{align*}
	{\Gc_{K^n} = \left\{ g_0^n \mbox{ H\"older continuous of any order  less than 1/2,}  \mbox{ satisfying } \eqref{Croissance} \mbox{ and } g^n_0(0) \geq 0  \right\}.}
	\end{align*} 
It is shown in \cite[Example 2.2]{AJEE18b} that the two specifications of input curves  \eqref{eq:g0 example 1}-\eqref{eq:g0 example 2} provided earlier satisfy \eqref{Croissance}.}\vs2

{We now provide the rigorous existence and uniqueness result for any initial input curve $g_0^n \in \Gc_{K^n}$. {We note that, for the specific choice \eqref{eq:g0 example 2}, the result is an immediate consequence of \cite[Theorem 7.1]{ALP17}.}}

\begin{theorem}[Existence and uniqueness]
		Fix $n \in \mathbb N$, $S_0^n >0$ and assume that $g_0^n \in \Gc_{K^n}$. Then, the stochastic differential equation \eqref{eq:liftedS}-\eqref{eq:liftedU} has a  unique  continuous strong solution $( S^n,V^n,(U^{n,i})_{1 \leq i \leq n})$ such that $V^n_t \geq 0$, for all $t\geq 0$, almost surely. Further, the process $S^n$ is a martingale.
\end{theorem}
{\begin{proof}
By virtue of the variation of constants formula on the factors, the  {lifted Heston model} is equivalent to  a Volterra Heston model in the sense of \cite{AJEE18b} of the form
	\begin{align}
	dS_t^n &=  S_t^n \sqrt{ V_t^n} dB_t , \quad  S_0^n> 0,  \label{E:Sn} \\
	V^n_t&= g^n_0(t)+  \int_0^t K^n(t-s) \left( -\lambda V^n_sds + \nu \sqrt{ V^n_s}dW_s \right), \label{E:Vn} 
	\end{align}
	with $K^n$ given by \eqref{eq:kernel Kn}.	 Since $K^n$ is locally Lipschitz and completely monotone, the assumptions of  \cite[Theorem 2.1]{AJEE18b} are met. Consequently, the stochastic Volterra equation \eqref{E:Sn}-\eqref{E:Vn} has a unique  $ \R_+^2$-valued weak continuous solution $(S^n,V^n)$ on some filtered probability space $(\Omega^n, \Fc^n, (\Fc^n_t)_{t \geq 0},\Q^n)$ for any initial condition $S_0^n >0$ and  {admissible input curve} $g^n_0 \in \Gc_{K^n}$. Moreover, since $K^n$ is differentiable, strong uniqueness  is ensured by \cite[Proposition B.3]{AJEE18a}. The claimed  existence and uniqueness statement now follows from \eqref{eq: volterra Un}. Finally, the martingality of $S^n$ follows along the lines of \cite[Theorem 7.1(iii)]{ALP17}.
\end{proof}}

{\subsection{The rough limiting model}}\label{A:limiting model}
We now discuss the convergence of the \textit{lifted Heston model} towards the rough Heston model \eqref{eq:rough price}-\eqref{eq:rough variance}, as the number of factors goes to infinity, we refer to \cite{AJEE18a} for more details. We fix $H \in (0,1/2)$ and we denote by  $K_H:t \to t^{H- \frac 12}/\Gamma(H+ 1/2)$ the fractional kernel of the rough Heston model appearing in \eqref{eq:rough variance}. The kernel $K_H$  can be re-expressed as a Laplace function  
		$$ K_H(t)  = \int_0^{\infty} e^{-xt}\mu(dx), \quad t \geq 0,$$
		with  $\mu(dx)=\frac{x^{-\alpha}}{\Gamma(\alpha)\Gamma(1-\alpha)}$ and $\alpha=H+1/2$. On the one hand, for a fixed $n$, the  parametrization \eqref{eq: ci and xi}	is linked to $\mu$ as follows:
	\begin{align}\label{eq:cxmu}
 c_i^n= \int_{\eta^n_{i-1}}^{\eta^n_{i}} \mu(dx), \quad x^n_i = \frac{1}{c_i^n}  \int_{\eta^n_{i-1}}^{\eta^n_{i}} \mu(dx), \quad i=1,\ldots,n, 
		\end{align}
		where $\eta^n_i=r_n^{i-n/2}$, for $i=0,\dots,n$.  We will show that, under \eqref{eq:rn cond},   
		\begin{align}\label{eq:conv L2}
		 K^n \to K_H, \quad \mbox{as $n$ goes to infinity,} \quad \mbox{in the } L^2 \mbox{ sense}.
		\end{align}
		On the other hand,  for each $n \in \mathbb N$, we have proved the existence of a solution to \eqref{eq:volterra Vn}. One would therefore expect from \eqref{eq:conv L2} the convergence of the sequence of solutions of \eqref{E:Vn} towards the solution of \eqref{eq:rough variance}.  This is indeed the case, as illustrated by the following theorem, which adapts  \cite[Theorem 3.5]{AJEE18a} to the geometric partition.

		\begin{theorem}[Convergence towards the rough Heston model]\label{T:convergence rHeston} 
				Consider a sequence $(r_n)_{n \geq 1}$ satisfying  $\eqref{eq:rn cond}$, and set $g_0^n$ as in \eqref{eq: flat curve} and  $(c^n_i,x^n_i)_{1 \leq i \leq n}$  as in \eqref{eq: ci and xi}, for every even $n=2p$, with $p \geq 1$. Assume $S_0^n=S_0$, for all $n$,  then, the sequence of solutions $(S^{n},V^{n})_{n=2p,p\geq 1}$  to \eqref{eq:liftedS}-\eqref{eq:liftedV} converges weakly, on the space of continuous functions on $[0,T]$ endowed with the uniform topology, towards the rough Heston model \eqref{eq:rough price}-\eqref{eq:rough variance}, for any $T>0$.
			\end{theorem}

	We will only sketch  the proof for the $L^2$ convergence of the kernels \eqref{eq:conv L2}, in order to highlight the small adjustments  that one needs to make to the proof of  \cite[Theorem 3.5]{AJEE18a}. Indeed, since $\eta_0^n \neq 0$ in our case, \cite[Theorem 3.5]{AJEE18a} cannot be  directly applied, compare with \cite[Assumption 3.1]{AJEE18a} where the left-end point of the partition is zero. 	The following lemma adapts \cite[Proposition 3.3]{AJEE18a} to the geometric partition. The rest of the proof of Theorem \ref{T:convergence rHeston}  follows along the lines of  \cite[Theorem 3.5]{AJEE18a}  by making the same small adjustments highlighted below, mainly to treat the integral chunk between $[0,\eta_0^n]$.  

	\begin{lemma}[Convergence of $K^n$ towards $K_H$]
			Let $(r_n)_{n \geq 1}$ as in  $\eqref{eq:rn cond}$, and  $(c^n_i,x^n_i)_{1 \leq i \leq n}$  given by \eqref{eq: ci and xi}. Define $K^n$ by \eqref{eq:kernel Kn}, then, 
			\begin{align}\label{eq:convL2}
			\|K^n - K_H\|_{L^2(0,T)} \to 0, \quad \mbox{as } n \to \infty,
			\end{align}  
		for all $T>0$.
	\end{lemma}
	
	\begin{proof}
		 Set $\eta^n_i=r_n^{i-n/2}$, for $i=0,\ldots,n$. Using \eqref{eq:cxmu},  we start by decomposing $(K_H-K^n)$ as follows 
		\begin{align*}
	K_H-K^n&=\int_0^{\infty}e^{-x(\cdot)}  \mu(dx)-\sum_{i=1}^n c^n_i e^{-x^n_i(\cdot)} \\
		&= \int_0^{\eta_0^n}e^{-x(\cdot)}  \mu(dx) + \left(\sum_{i=1}^n \int_{\eta_{i-1}^n}^{\eta_i^n}\left(e^{-x(\cdot)} - e^{-x^n_i(\cdot)}\right)\mu(dx) \right) + \int_{\eta_n^n}^{\infty} e^{-x(\cdot)}  \mu(dx)\\
		&:= J_1^n + J_2^n + J_3^n,
		\end{align*}
		so that 
		\begin{align*}
		\|K_H-K^n\|_{L^2(0,T)} \leq I_1^n + I_2^n + I_3^n,
		\end{align*}
		with $I^n_k=\|J^n_k\|_{L^2(0,T)}$, for $k=1,2,3$. We now prove that each $I^n_{k} \to 0$, as $n$ tends to $\infty$.   	
		Relying on a second order Taylor expansion, along the lines of the proof of \cite[Proposition 7.1]{coutin2007approximation}, we get the following bound 
		$$\left |  \int_{\eta_{i-1}^n}^{\eta^n_i} \left( e^{-x t}-e^{-x_i^n t}\right) \mu(dx) \right | \leq C \, t^2 \,  r_n^{1/2} \, (r_n-1)^2 \, \int_{\eta_{i-1}^n}^{\eta^n_i}  (1 \wedge x^{-1/2}) \mu(dx) , \quad t \leq T,$$
		for all $i=1,\ldots,n$, where $C$ is a constant independent of $n$, $i$ and $t$. Summation over $i=1,\ldots,n$ leads to 
		$$ I^n_{2} \leq C \, \frac{T^{5/2}}{\sqrt{5}}   \,  r_n^{1/2} \, (r_n-1)^2 \, \int_{0}^{\infty}  (1 \wedge x^{-1/2}) \mu(dx),$$
		so that $I^n_2 \to 0 $, as $n \to \infty$, by virtue of the first condition in \eqref{eq:rn cond}. On another note, 
		$$ I_1^n \leq \int_0^{\eta_0^n} \mu(dx)=\frac{(\eta_0^n)^{1-\alpha}}{\Gamma(\alpha)\Gamma(2-\alpha)}= \frac{r_n^{-(1-\alpha)n/2}}{\Gamma(\alpha)\Gamma(2-\alpha)} \to 0, \quad \mbox{when } n \to \infty,$$
		thanks to the second condition in \eqref{eq:rn cond}. Similarly, 
		$$ {I_3^n \leq \int_{\eta_n^n}^{\infty} \sqrt{\frac{1-e^{-2xT}}{2x}} \mu(dx)  \leq  \frac{r_n^{(1/2-\alpha)n/2}}{\Gamma(\alpha)\Gamma(1-\alpha)(\alpha-1/2)} \to 0, \quad \mbox{when } n \to \infty.} $$
 		Combining the above leads to \eqref{eq:convL2}.
		\end{proof}

\section{The full Fourier-Laplace transform}\label{S: appendix full fourier}

We provide the full Fourier-Laplace transform for the joint process $X^n:=(\log S^n,V^n)$ extending \eqref{eq:char function log S}. The formula can be used to price path-dependent options on the stock price $S^n$ and the variance process $V^n$.\vs2

{Once again, this is a particular case of  \cite[Section 4]{AJEE18b}, by observing that $K^n$ defined in \eqref{eq:kernel Kn} is the Laplace transform of the following nonnegative measure 
$$\mu^n(dx)= \sum_{i=1}^n c^n_i \delta_{x^n_i}(dx) .$$}

 Fix row vectors $u=(u_1,u_2) \in \mathbb C^2$ and $f\in L^1_{\rm loc}(\R_+,(\mathbb C^2))$ such that
\begin{align*}
\text{${\rm Re\,} (u_1+1*f_1) \in[0,1]$, \; ${\rm Re\,} u_2 \le0$ \; and \; ${\rm Re\,} f_2\le0$,}
\end{align*}
then,  {it follows from \cite[Remark 4.3]{AJEE18b} with $\mu=\sum_{i=1}^n c_i^n \delta_{x^n_i}$} that the  Fourier-Laplace transform of $X^n=(\log S^n, V^n)$ is exponentially affine with respect to the family $(U^{n,i})_{ 1 \leq i \leq n}$, 
\begin{align*}
\E\left[ \exp\left(u X^n_T + (f*X^n)_T\right) \Mid \Fc_t \right] = \exp\left({\phi^n(t,T)  +  \psi_1(T-t) \log S^n_t  +  \sum_{i=1}^n  c^n_i\psi_2^{n,i} (T-t) U^{n,i}_t }\right),
\end{align*}
for all $t \leq T$, where  
 $(\psi_1,(\psi^{n,i}_2)_{1 \leq i \leq n})$ are  the unique solutions of the  following  system of Riccati ordinary differential equations
\begin{align*}
\psi_1 &= u_1+1*f_1, \\
(\psi^{n,i}_2)' &= - x^n_i  \psi^{n,i}_2 + F\left(\psi_1,   \sum_{j=1}^n c^n_j \psi^{n,j}_2 \right),\quad  \psi^{n,i}_2(0)=u_2, \quad i=1,\ldots, n,
\end{align*}
with 
\begin{align*}
F(\psi_1,\psi_2) &= f_2+\frac12\left( \psi_1^2-\psi_1\right) +(\rho \nu \psi_1- \lambda)  \psi_2 + \frac{\nu^2}{2} \psi^2_2
\end{align*}
and
\begin{align*}
\phi^n(t,T) &=  u_2 g^n_0(T ) + \int_0^{T-t} F\left(\psi_1,\sum_{i=1}^n  c^n_i\psi_2^{n,i}(s)\right)  g^n_0(T-s) ds  + \int_0^t f(T-s)X_s ds , \quad t \leq T.
\end{align*}

\section{Discretization schemes}\label{A:scheme}

\subsection{Riccati equations}
The aim of this section is to design an approximation scheme of the $n$-dimensional Riccati system of equations \eqref{eq:psi xi}. In order to gain some insights, consider first the case where $F \equiv 0$  so that \eqref{eq:psi xi} reduces to 
\begin{align}\label{eq:Riccati F=0}
 (\psi^{n,i})' = - x^n_i \psi^{n,i} , \quad  i=1, \ldots , n,
\end{align}
and the solution is given by 
$$  \psi^{n,i}(t)= \psi^{n,i}(0) e^{-x^n_i t} , \quad  i=1, \ldots , n.$$
  One could start with an explicit Euler scheme for \eqref{eq:Riccati F=0}, that is 
$$ \hat \psi^{n,i}_{t_{k+1}} = \hat \psi^{n,i}_{t_{k}}  - x^n_i \Delta t \hat \psi^{n,i}_{t_k} =  \left(1-x^n_i \Delta t \right ) \hat \psi^{n,i}_{t_{k}},  \quad  i=1, \ldots , n,  $$ 
for a regular time grid $t_k=(kT)/N$ for all $k= 1, \ldots , N$, where $T$ is the terminal time, $N$ the number of time steps and $\Delta t = T/N$.  A sufficient condition for the stability  of the scheme reads 
$$  \Delta t   \leq  \min_{1 \leq i \leq n} \frac{1}{x^{n}_i}.  $$
Recall from \eqref{eq: ci and xi} that $x^n_n$ grows very large as  $n$ increases. For instance, for $n=20$, $r_{20}=2.5$ and $H=0.1$, $x^n_n= 6417.74$. Consequently, if one needs to ensure the stability of the explicit scheme, one needs a very large number of time steps $N$. In contrast, the implicit Euler scheme 
$$  \hat \psi^{n,i}_{t_{k+1}} = \hat \psi^{n,i}_{t_{k}}  - x^n_i \Delta t \, \hat \psi^{n,i}_{t_{k+1}},  \quad  i=1, \ldots , n,   $$
is stable for any number of time steps $N$ and reads
$$ \hat \psi^{n,i}_{t_{k+1}} =  \frac{1}{1+ x^n_i \Delta t  }\hat \psi^{n,i}_{t_{k}},  \quad  i=1, \ldots , n.$$   
For this reason, we consider the following explicit-implicit discretization scheme of the $n$-dimensional Riccati system of equations \eqref{eq:psi xi}
\begin{align}\label{eq:explicit implicit riccati}
\hat \psi^{n,i}_0 &= 0, \quad  
\hat \psi^{n,i}_{t_{k+1}} =   \frac{1}{1+ x^n_i \Delta t }  \left(\hat \psi^{n,i}_{t_{k}}  + \Delta t \, F\left(u,   \sum_{j=1}^n c^n_j \hat \psi^{n,j}_{t_{k}} \right) \right), \quad i=1,\ldots, n,
\end{align}
for a regular time grid $t_k=k \Delta t$ for all $k= 1, \ldots , N$, with time step size $\Delta t = T/N$,  terminal time $T$ and number of time steps $N$. 
Alternatively, one could also consider the exponential scheme for the Riccati equations  by replacing the term $1/(1+x^n_i \Delta t)$ with $e^{-x^n_i \Delta t}$. One can also combine more involved discretization schemes for the explicit part involving the quadratic function $F$, for instance higher order Runge-Kutta methods can be used, see \cite{lambert1991numerical}.

\subsection{Stochastic process}\label{S: sim Vn}
Similarly, we suggest to consider the following modified explicit-implicit scheme for the variance process $V^n$:
\begin{align}
\hat V^n_{t_k} &= g_0^n(t_k) + \sum_{i=1}^n c^n_i \hat U^{n,i}_{t_k}, \quad \;  \hat U^{n,i}_0 =0, \label{eq:simliftedV}\\
\hat U_{t_{k+1}}^{n,i} &= \frac{1}{1+ x^n_i \Delta t }  \left(\hat U_{t_{k}}^{n,i} -\lambda \hat V^n_{t_k}  \Delta t + \nu \sqrt{\left(\hat V^n_{t_k}\right)^+} \left(W_{t_{k+1}}-W_{t_k}\right)\right), \quad i=1,\ldots,n, \label{eq:simliftedU}
\end{align}
for a regular time grid $t_k=k \Delta t$, $k=1 \ldots N$,  $\Delta t = T/N$ and  {$ (W_{t_{k+1}}-W_{t_k}) \sim \mathcal N (0,\Delta t)$.} Notice that we take the positive part $(\cdot)^+$ since the simulated process can become negative. Once there, simulating  the spot-price process $S^n$ is straightforward. We leave the theoretical study of  convergence and stability for future work. {Numerically, the scheme seems stable.} 
 Alternatively,  one could also consider the exponential scheme for the stochastic process by replacing the term $1/(1+x^n_i \Delta t)$ with $e^{-x^n_i \Delta t}$. As a final remark, one notices that \eqref{eq:simliftedV}-\eqref{eq:simliftedU} corresponds to the space-time discretization of the integro-differential stochastic partial differential equation \eqref{eq:spde U1}-\eqref{eq:spde U2}. This is illustrated on Figure \ref{fig:spde} below.

 \newpage
 ~
 \vs2\vs2\vs2\vs2\vs2\vs2\vs2\vs2\vs2\vs2\vs2\vs2\vs2\vs2\vs2\vs2\vs2\vs2\vs2\vs2
\begin{center}
	\includegraphics[scale=0.65]{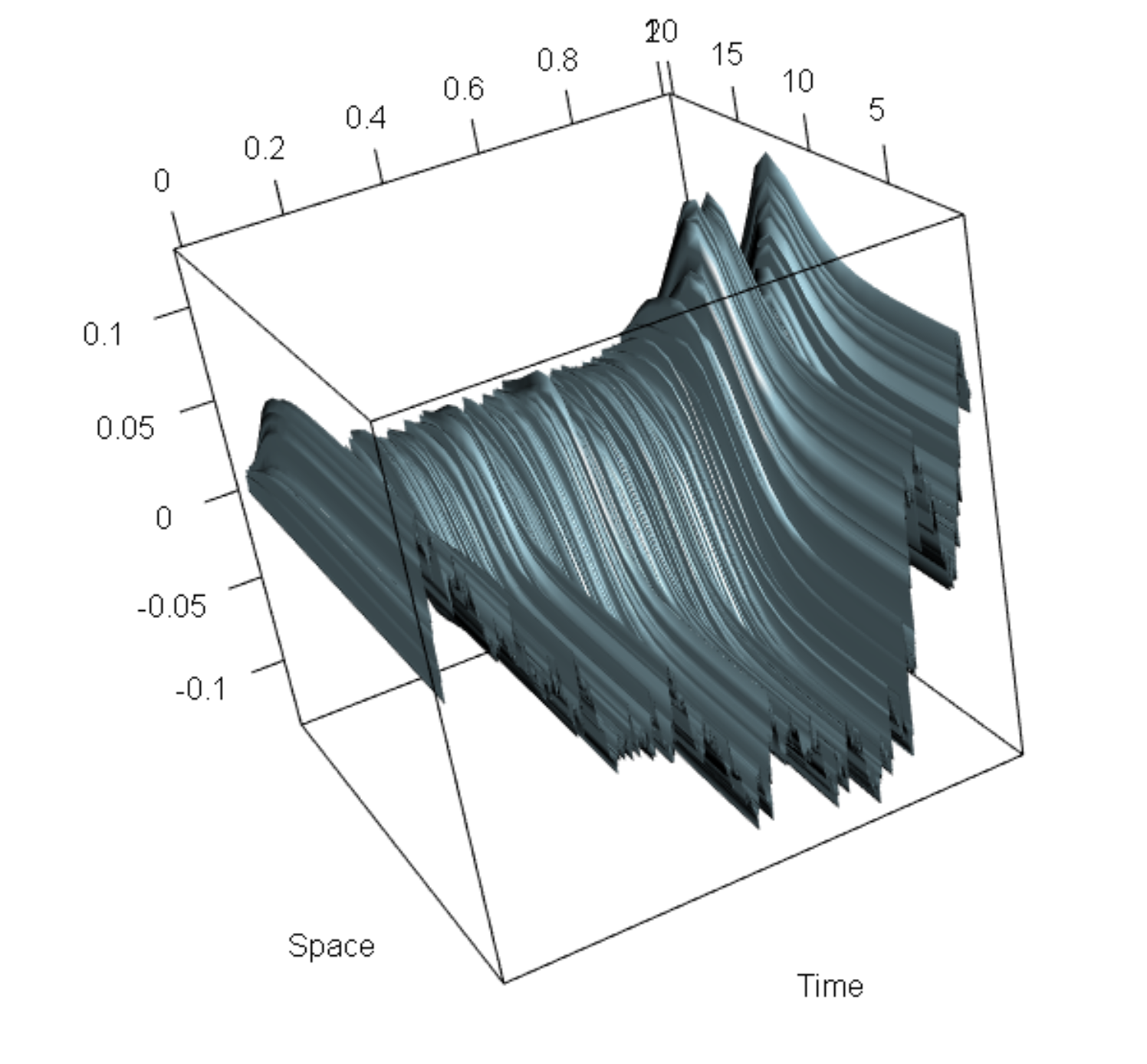}
	\rule{35em}{0.5pt}
	\captionof{figure}{{Simulated path of  the stochastic partial differential equation \eqref{eq:spde U1}-\eqref{eq:spde U2} by using the scheme \eqref{eq:simliftedV}-\eqref{eq:simliftedU}.}}
	\label{fig:spde}
\end{center}


\newpage
~
\vs2\vs2\vs2\vs2\vs2\vs2\vs2\vs2\vs2\vs2\vs2\vs2\vs2\vs2
\begin{center}
	\includegraphics[scale=0.65]{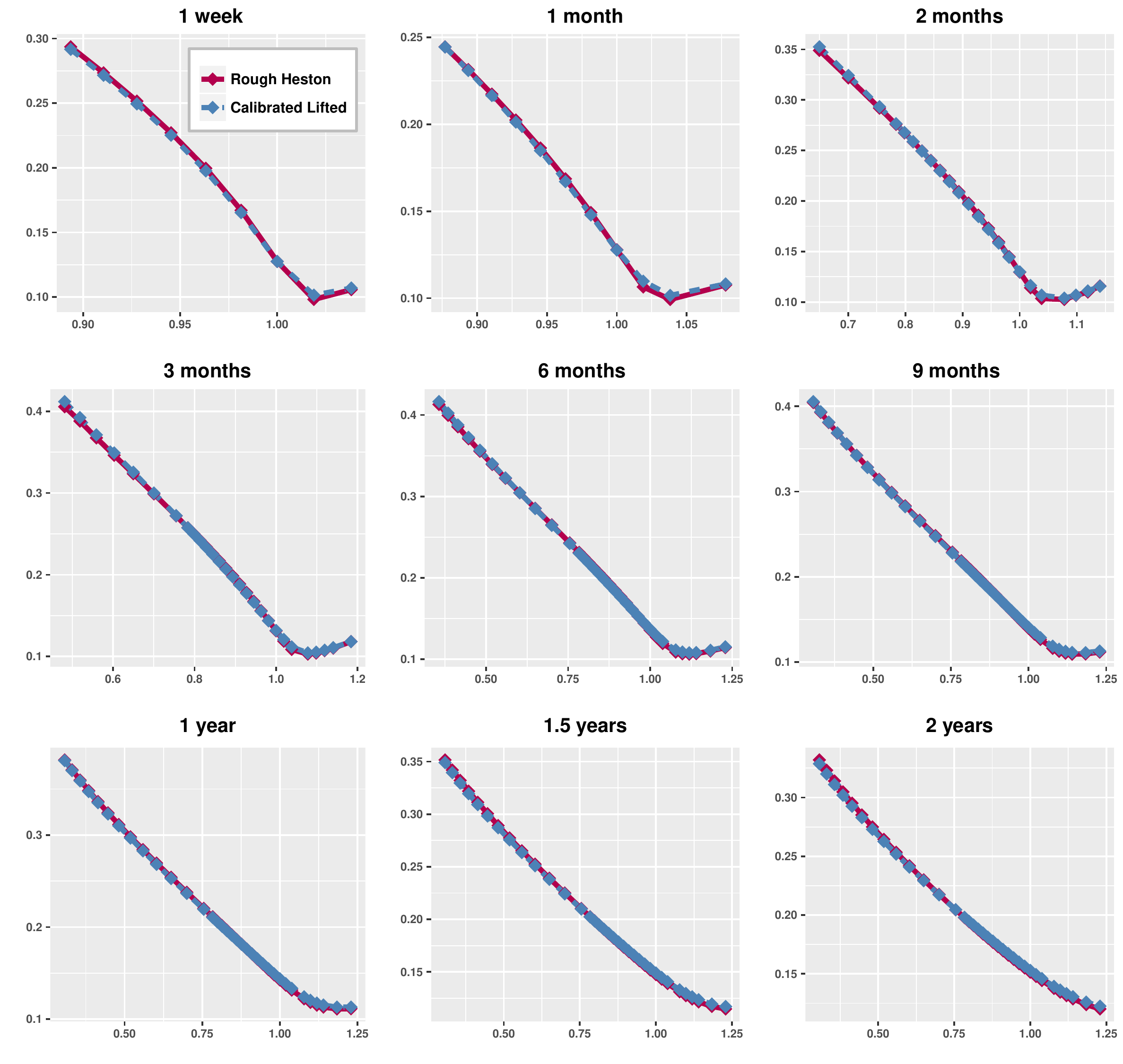}
	\rule{35em}{0.5pt}
	\captionof{figure}{{{Implied volatility surface of the rough Heston model $\sigma_{\infty}(K, T;\Theta_0)$  of \eqref{eq: 5parameters} (red) and the calibrated {lifted Heston model} $\sigma_{20}(K, T; r_{20}=2.5, \hat\Theta_0)$ of Table \ref{table:calibrate lifted}  (blue) for maturities ranging from $1$ week to $2$ years ($\mbox{MSE}=4.01\mbox{e-}07)$.}}}
	\label{fig:lifted vs rough}
\end{center}

\newpage
~
\vs2\vs2\vs2\vs2\vs2\vs2\vs2\vs2\vs2\vs2\vs2\vs2\vs2\vs2
\begin{center}
	\includegraphics[scale=0.65]{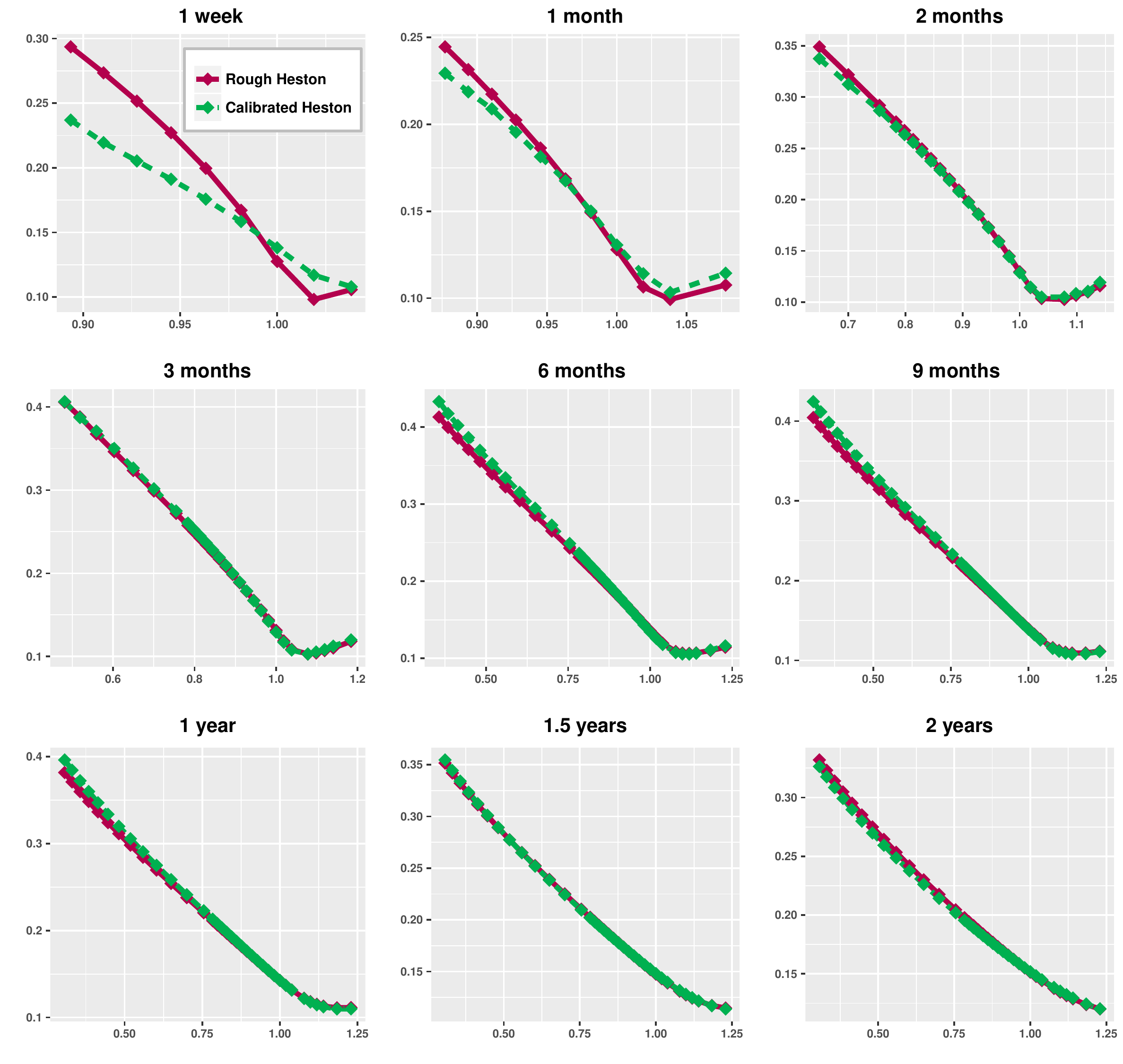}
	\rule{35em}{0.5pt}
	\captionof{figure}{{Implied volatility surface of the rough Heston model $\sigma_{\infty}(K, T;\Theta_0)$ (red) and the calibrated Heston model of Table \ref{table:calibrate Heston} (green) for maturities ranging from $1$ week to $2$ years ($\mbox{MSE}=$ 2.06e{-03}).}}
	\label{fig:calibrated heston surface}
\end{center}

\newpage

~
\vs2\vs2\vs2\vs2\vs2\vs2\vs2\vs2\vs2\vs2\vs2\vs2\vs2\vs2
\vs2\vs2\vs2\vs2\vs2\vs2\vs2\vs2\vs2\vs2\vs2

\begin{table}[h!]
	\centering  
	\begin{tabular}{ccccc} 
		\hline\hline                        
  $\nu$ & $\rho$ & $H$ & & MSE \\  
		\hline
0.22 & -0.67 & 0.09 &  & 3.63e-06 \\ 
0.14 & -0.54 & 0.19 &  & 5.34e-06 \\ 
0.35 & -0.65 & 0.19 &  & 8.17e-06 \\ 
0.14 & -0.83 & 0.06 &  & 9.74e-05 \\ 
0.22 & -0.59 & 0.15 &  & 4.60e-06 \\ 
0.37 & -0.50 & 0.12 &  & 4.55e-06 \\ 
0.40 & -0.53 & 0.11 &  & 4.56e-06 \\ 
0.34 & -0.85 & 0.08 &  & 3.45e-04 \\ 
0.22 & -0.89 & 0.09 &  & 1.25e-04 \\ 
0.44 & -0.76 & 0.11 &  & 2.79e-04 \\ 
0.32 & -0.70 & 0.12 &  & 4.56e-06 \\ 
0.42 & -0.63 & 0.08 &  & 5.22e-06 \\ 
0.10 & -0.61 & 0.17 &  & 3.69e-06 \\ 
0.42 & -0.64 & 0.11 &  & 4.81e-06 \\ 
0.30 & -0.69 & 0.17 &  & 5.96e-06 \\ 
0.06 & -0.71 & 0.17 &  & 2.98e-06 \\ 
0.36 & -0.71 & 0.16 &  & 6.14e-06 \\ 
0.25 & -0.80 & 0.18 &  & 1.63e-04 \\ 
0.09 & -0.77 & 0.06 &  & 2.87e-06 \\ 
0.35 & -0.74 & 0.13 &  & 1.44e-04 \\ 

	\hline
	\end{tabular}
	\caption{{\textbf{Robustness of $r_{20}=2.5$:} First $20$ values of the simulated parameters and the corresponding mean squared error between the implied volatility surface of the lifted model $\sigma_{20}(K,T;2.5,\Theta_k)$ and the rough model $\sigma_{\infty}(K,T;\Theta_k)$, for $k=1,\ldots,20$.}}
	\label{table random20} 
\end{table}

 \end{appendices}

 \newpage
\bibliographystyle{plainnat}
  {\bibliography{biblifting}}
 
\end{document}